\documentclass[a4paper,11pt]{article}
\usepackage[utf8]{inputenc}

\usepackage{geometry}
\usepackage{xcolor}
\usepackage{amsmath, array, amssymb, amsfonts,amsthm}
\usepackage{xfrac}
\usepackage[inline]{enumitem}

\usepackage[french, english]{babel}
\usepackage[all]{xy}

\usepackage{txfonts}  

\usepackage{sectsty}
\sectionfont{\centering}
\subsectionfont{\centering}
\subsubsectionfont{\centering}

\usepackage{booktabs}
\usepackage{caption}
\usepackage{dsfont}
\usepackage{mathtools}

\usepackage[hidelinks]{hyperref}

\usepackage{textcomp}

\usepackage{multicol}
\usepackage[square,numbers,sort,compress,semicolon,merge]{natbib}
\let\cite\citep 
\usepackage{hyperref}
\hypersetup{colorlinks=true, urlcolor=blue, citecolor=blue, linktoc=page}

\makeatletter
\renewcommand*\env@matrix[1][\arraystretch]{%
  \edef\arraystretch{#1}%
  \hskip -\arraycolsep
  \let\@ifnextchar\new@ifnextchar
  \array{*\c@MaxMatrixCols c}}
\makeatother

\renewcommand\P{\mathcal{P}}
\newcommand\M{\mathcal{M}}

\newcommand\RR{\mathbb{R}}
\newcommand\CC{\mathbb{C}}
\renewcommand\1{\textbf{1}}
\newcommand\id{\textit{id}}

\renewcommand\H{\mathcal{H}}

\renewcommand\S{\mathcal{S}}

\newcommand\U{\mathcal{U}}

\newcommand\SL{\mathcal{SL}}

\newcommand\K{\mathcal{K}}
\newcommand\J{\mathcal{J}}

\newcommand\W{\mathcal{W}}
\newcommand\vphi{\varphi}

\renewcommand\l{\text{\tiny{L}}}
\newcommand\ww{\text{\tiny{W}}}
\renewcommand\ss{\text{s}}
\newcommand\n{\text{\tiny{N}}}
\newcommand\sP{\mathsf{P}}
\newcommand\sC{\mathsf{C}}
\newcommand\sT{\mathsf{T}}
\newcommand\sW{\mathsf{W}}
\newcommand\sS{\mathsf{S}}
\newcommand\sM{\mathsf{M}}
\newcommand\sE{\mathsf{E}}

\renewcommand\epsilon{\varepsilon}

\newcommand\rarrow{\rightarrow}

\newcommand\LieG{\mathfrak{g}}
\newcommand\LieH{\mathfrak{h}}

\newcommand\su{\mathfrak{su}}
\newcommand\so{\mathfrak{so}}
\renewcommand\sl{\mathfrak{sl}}
\newcommand\co{\mathfrak{co}}

\newcommand\LieK{\mathfrak{k}}
\newcommand\LieJ{\mathfrak{j}}

\renewcommand\t{\widetilde}
\newcommand\h{\widehat}
\renewcommand\b{\bar }
\newcommand\w{\wedge}
\renewcommand\d{\partial}
\newcommand\s{\sigma}

\renewcommand\-{^{-1}}
\newcommand\Ad{\text{Ad}}
\newcommand\ad{\text{ad}}
\renewcommand\id{\text{id}}
\renewcommand\1{\mathds{1}}

\DeclareMathOperator{\Aut}{Aut}
\DeclareMathOperator{\Tr}{Tr}

\newtheorem{thm}{Theorem}

\newtheorem{cor}[thm]{Corollary}
\newtheorem{prop}[thm]{Proposition}
\theoremstyle{definition}

\hypersetup{
pdfauthor={J. Attard, J. François},
pdftitle={Tractors and local Twistors from conformal Cartan geometry: a gauge theoretic approach},
pdfsubject={},
pdfcreator={pdflatex},
pdfproducer={pdflatex},
pdfkeywords={}
}

\begin{document}

\title{Tractors and Twistors from conformal Cartan geometry: a gauge theoretic approach \\ 
\center II. Twistors}
\author{J. Attard,${\,}^a$ \hskip 0.7mm J. François }
\date{}

\maketitle
\begin{center}

\vskip -0.8cm
${}^a$ Aix Marseille Univ, Université de Toulon, CNRS, CPT, Marseille, France
\end{center}

\begin{abstract}
Tractor and Twistor bundles both provide natural conformally covariant calculi on $4D$-Riemannian manifolds. They have different origins but are closely related, and usually constructed bottom-up through prolongation of defining differential equations. We propose alternative top-down gauge theoretic constructions, starting from the conformal Cartan bundle $\P$ and its vectorial $E$ and spinorial $\sE$ associated bundles.
Our key ingredient is the dressing field method of gauge symmetry reduction, which allows to exhibit tractors and twistors and their associated connections as gauge fields of a non-standard kind as far as Weyl rescaling symmetry is concerned. By \emph{non-standard} we mean that they implement the gauge principle of physics, but are of a different geometric nature than the well known differential geometric objects usually underlying gauge theories.   
We provide the corresponding BRST treatment. In a companion paper we dealt with tractors, in the present one we address the case of twistors. 
\end{abstract}






\tableofcontents

\section{Introduction}  

Twistor theory probably needs no introduction. Let us just remind that it was devised by Penrose in
the $60$'s and early $70$'s as an alternative framework for physics - and quantum gravity - in which conformal symmetry is pivotal \cite{Penrose1960, Penrose1967, Penrose1968, Penrose-McCallum72}. In twistor theory spinors quantities takes over the role of tensors and the - conformally compactified - Minkowski space is seen as secondary, emerging from a more fundamental twistor space hoped to be more fit for quantization.  A generalization to arbitrary pseudo-Riemannian manifolds gave rise to the concept of local twistors, which provides a conformal spinorial calculus. 
The standard reference text is Penrose and Rindler \cite{Penrose-Rindler-vol1, Penrose-Rindler-vol2}.
Twistor theory remains an active area of research in physics.

Tractors are perhaps not as well known from physicists, but more so from mathematicians. 
In the period of renewal of differential geometry sparked by Einstein's General Relativity, around $1922$-$1926$ Cartan developed the notions of moving frames and \emph{espace généralisés}. These are manifolds with torsion in addition to curvature, classic examples being manifolds endowed with projective and conformal connections. Cartan's work  became widespread and lead to the further development by Whitney and Ehresmann of connections on fibered manifolds, known to be the underlying geometry of Yang-Mills theories.
 In $1925$-$26$, Thomas developed, independently from Cartan,  a calculus on torsionless conformal (and projective) manifolds, analogous to Ricci calculus on Riemannian manifolds.  His work was rediscovered and expanded in \cite{Bailey-et-al94} where it was given its modern guise as a vector bundle  called \emph{standard tractor bundle} endowed with a linear connection, the \emph{tractor connection}. In recent years this conformal tractor calculus has been of interest for physicists, see e.g \cite{Gover-Shaukat-Waldron09, Gover-Shaukat-Waldron09-2}
\medskip

Tractors and twistors are closely linked. It has been noticed that both define vector bundles associated to the conformal Cartan principal bundle $\P(\M, H)$ - with $H$ the parabolic subgroup of the conformal group $S\!O(2, 4)$
\footnote{First introduced to physics, according to \cite{Fulton_et_al1962}, by \cite{Cunningham} and \cite{Bateman1909, Bateman1910} in connection with Special Relativity and the invariance of Maxwell's equations.} 
 comprising Lorentz, Weyl and conformal boost symmetries - and that tractor and twistor connections are induced by the so-called \emph{normal Cartan connection} $\varpi_\n$ on $\P$. 
 
In standard presentations however, both tractors and local twistors are constructed through the prolongation of defining differential equations defined on a Riemannian manifold $(\M, g)$: the Almost Einstein equation and Twistor Equation respectively. The systems thus obtained are linear and closed, so that they can be rewritten as  linear operators acting on multiplets of variables called parallel tractors and global twistors respectively. The  behavior of the latter under Weyl rescaling of the underlying metric is given by definition and commutes with the actions of their associated  linear operators, which are then respectively called tractor and twistor connections. The multiplets are then seen as parallel sections  of  vectors bundles, the tractor and local twistor bundles, endowed with their linear connections. 
For the procedure in the tractor case see \cite{Bailey-et-al94}, or \cite{Curry-Gover2015} for a more recent and detailed review. For the twistor case see the classic \cite{Penrose-Rindler-vol2}, or \cite{Bailey-Eastwood91} which generalizes the twistor construction to paraconformal (PCF) manifolds. 
This constructive procedure via prolongation has been deemed more explicit \cite{Bailey-Eastwood91}, more intuitive and direct \cite{Bailey-et-al94} than the viewpoint in terms of vector bundles associated to $\P(\M, H)$. Since it starts from $(\M, g)$ to built  a gauge structure on top of it - vector bundles endowed with connections - we may call-it a ``bottom-up'' approach.  
\medskip

In this paper and its companion we put forward a ``top-down'' approach to tractors and twistors that relies on a gauge theoretic method of gauge symmetry reduction: the dressing field method. Given a gauge structure (fiber bundles with connections) on $\M$, this method allows to systematically construct partially gauge-invariant composite fields built from the usual gauge fields and a so-called \emph{dressing field}. According to the transformations of the latter under  residual gauge symmetries, the composite fields  display interesting properties. In a noticeable case they are actually gauge fields of non-standard kind, meaning that they implement the gauge principle but are not of the same geometric nature as the usual Yang-Mills fields. The  method fits in the BRST framework. 

The common gauge structure on $\M$ we start with is the conformal Cartan bundle $\P(\M, H)$ endowed with a Cartan connection. If we add the vector bundle $E$ associated to the defining representation $\RR^6$ of $H$, the dressing approach allows to erase the conformal boost symmetry and to recover tractors and the tractor connection. This has been detailled in \cite{Attard-Francois2016}. In this paper, we consider the vector bundle  associated to the spin representation $\CC^4$ of $S\!O(2, 4)$ and $H$. From its dressing we shall obtain tractors and a generalized twistor connection induced by the spin representation of the Cartan connection. When the normal Cartan connection is considered, we recover the standard twistor connection. We stress that  twistors thus obtained, while being genuine standard gauge fields with respect to (w.r.t) Lorentz gauge symmetry, are examples of non-standard gauge fields alluded to above w.r.t Weyl gauge symmetry. This, we think, is a new consideration worth emphasizing. 
\medskip

We don't want to force the reader to skip through the companion paper  \cite{Attard-Francois2016} to find definitions and notations, much less to gather results. To the advantage of the reader, we would rather  make the present paper as self-contained as possible. But to do so we are bound to duplicate significant background material in the first two sections. 
Then, in section \ref{The geometry of gauge fields} we review the basics of differential geometry underlying gauge theories, including Cartan geometry, as well as the BRST formalism, so as to fix notations and define important notions. 

In section \ref{Reduction of gauge symmetries: the dressing field method} we review the dressing field method of gauge symmetry reduction. We provide a number of general propositions, proofs of  which can be found in the corresponding section of \cite{Attard-Francois2016}. We emphasis, as a result of the method,  the possible emergence of gauge fields of a non-standard kind which we characterize. 

Finally,  section \ref{Twistors from conformal Cartan geometry via dressing} starts with a brief review of the ``bottom-up'' procedure for twistors. Then we describe the conformal Cartan bundle as well as the group morphism $H \rarrow \b H\subset S\!U(2, 2)$ and the Lie algebra morphism $\so(2, 4) \rarrow \su(2, 2)$,  rediscovering for ourselves  results of  \cite{Klotz74}. At last, we put this material to use in constructing twistors and the twistor connection, ``top-down'' through  dressing. In doing so we reproduce and generalize results of  \cite{Friedrich77} connecting the normal Cartan connection and the twistor connection. Residual Lorentz and Weyl gauge symmetries are analyzed both at the finite and BRST level. Finally, we give a geometrically clear way to recover the results of \cite{Merkulov1984_I} concerning a twistorial approach to Weyl gravity, and formulate a critic of its unification with electromagnetism as proposed in \cite{Merkulov1984_II}. 
We summarize our results and gather our comments  in our conclusion \ref{Conclusion}.

\section{The geometry of gauge fields}   
\label{The geometry of gauge fields}   

Gauge theories are a cornerstone of modern physics built on the principle that the fundamental interactions originate from local symmetries called gauge symmetries. 
The mathematics underlying classical gauge theories is now widely known to be the differential geometry of fiber bundles and connections supplemented by the differential algebraic BRST approach.
In order to fix notations, we briefly recall the basic features of these in this section.

\subsection{Basic differential geometry} 
\label{Basic differential geometry} 

Let $\P(\M, H)$ be a principal fiber bundle over a smooth $n$-dimensional manifold $\M$, with structure Lie group $H$ and projection map $\pi: \P \rarrow \M$.  Given a representation $(\rho, V)$ for $H$ we have the associated bundle $E:=P\times_\rho V$, whose sections are in bijective correspondence with $\rho$-equivariant maps on $\P$: $\t\vphi \in \Gamma(E) \leftrightarrow \vphi \in \Lambda^0(\P, \rho)$.

Given the right-action  $R_hp=ph$ of $H$ on $\P$, a  $V$-valued $n$-form $\beta$ is  said $\rho$-equivariant iff $R^*_h\beta=\rho(h\-)\beta$. Let $X^v\in V\P\subset T\P$ be a vertical vector field induced by the infinitesimal action of $X \in \LieH=$ Lie$H$ on $\P$. A form $\beta$ is said horizontal if $\beta(X^v, \ldots)=0$. A form $\beta$ is said $(\rho, V)$-tensorial if it is both horizontal and $\rho$-equivariant. 
\smallskip

Let $\omega \in \Lambda^1(\P, \LieH)$ be a choice of connection on $\P$: it is $\Ad$-equivariant and satisfies $\omega(X^v)=X$. The horizontal subbundle $H\P \subset T\P$, the non-canonical complement of $V\P$, is defined as  $\ker \omega$. Given $Y^h \in H\P$ the horizontal projection of a vector field $Y \in T\P$, the covariant derivative of a $p$-form $\alpha$ is defined by $D\alpha:=d\alpha(Y_1^h, \ldots, Y_p^h)$. 

The connection's curvature form $\Omega \in \Lambda^2(\P, \LieH)$  is defined as its covariant derivative, but is algebraically given by the Cartan structure equation $\Omega=d\omega+\tfrac{1}{2}[\omega, \omega]$. 
Given a $(\rho, V)$-tensorial $p$-forms $\beta$ on $\P$,  its covariant derivative  is a $(\rho, V)$-tensorial $(p+1)$-form  algebraically given by $D\beta= d\beta +\rho_*(\omega)\beta$.  Furthermore  $D^2\beta=\rho_*(\Omega)\beta$. 

Since $\vphi$ is a $(\rho, V)$-tensorial $0$-form, its covariant derivative is the $(\rho, V)$-tensorial $1$-form $D\vphi:=d\vphi + \rho_*(\omega)\vphi$. The section $\vphi$ is said parallel if $D\vphi=0$. One can show that the curvature $\Omega$ is a $(\Ad, \LieH)$-tensorial $2$-form, so its covariant derivative is $D\Omega=d\Omega +\ad(\omega)\Omega=d\Omega +[\omega, \Omega]$. Given the Cartan structure equation, this vanishes identically and provides the Bianchi identity $D\Omega=0$.

Given a local section $\s: \U \subset \M \rarrow \P$, we have that $\s^*\omega \in \Lambda^1(\U, \LieH)$ is a Yang-Mills gauge potential, $\s^*\Omega \in \Lambda^2(\U, \LieH)$ is the Yang-Mills field strength and $\s^*\vphi$ is a matter field, while $\s^*D\vphi=d\s^*\vphi+\rho_*(\s^*\omega)\s^*\vphi$ is the minimal coupling of the matter field to the gauge potential. 
\smallskip

The group of vertical automorphisms of $\P$, $\Aut_v(\P):=\left\{\Phi:\P \rarrow \P \ |\  h\in H, \Phi(ph)=\Phi(p)h \text{ and } \pi \circ \Phi= \Phi \right\}$ is isomorphic to the gauge group $\H:=\left\{ \gamma :\P \rarrow H\ | \  R^*_h\gamma(p)=h\- \gamma(p) h  \right\}$, the isomorphism being $\Phi(p)=p\gamma(p)$. The composition law of $\Aut_v(\P)$, $\Phi_2^*\Phi_1:=\Phi_1 \circ \Phi_2$, implies that the gauge group acts on itself by $\gamma_1^{\gamma_2}:=\gamma_2^{-1} \gamma_1 \gamma_2$. 

The gauge group $\H \simeq \Aut_v(\P)$ acts on the connection, curvature and $(\rho, V)$-tensorial forms as,
\begin{align}
\label{ActiveGT}
&\omega^\gamma:=\Phi^*\omega=\gamma\-\omega\gamma + \gamma\- d\gamma, \quad \Omega^\gamma:=\Phi^*\Omega=\gamma\-\Omega \gamma, \\ 
&\vphi^\gamma:= \Phi^*\vphi=\rho(\gamma\-)\vphi, \quad  \text{and} \quad (D\vphi)^\gamma=D^\gamma \vphi^\gamma=\Phi^*D\vphi=\rho(\gamma\-)D\vphi.\notag
\end{align}
These are \emph{active} gauge transformations, formally identical but to be conceptually distinguished from \emph{passive} gauge transformations relating two local descriptions of the same global objects. Given two local sections related via  $\s_2=\s_1 h$, either over the same open set $\U$ of $\M$ or over the overlap of two open sets $\U_1 \cap \U_2$, one finds
\begin{align}
\label{PassiveGT}
&\s_2^*\omega=h\-\s_1^*\omega\  h + h\- dh, \quad \s_2^*\Omega=h\-\s_1^*\Omega\  h, \\ 
&\s_2^*\vphi= \rho(h\-)\s_1^*\vphi, \quad  \text{and} \quad \s_2^*D\vphi=\rho(h\-)\s_1^*D\vphi.\notag
\end{align}
This distinction active vs passive gauge transformations is reminiscent of the distinction diffeomorphisms vs coordinate transformations in General Relativity. 
\smallskip

If the manifold is equipped with a $(r, s)$-Lorentzian metric allowing for a Hodge star operator, and if $V$ is equipped with an inner product $\langle\ , \rangle$, then the prototypical Yang-Mills Lagrangian $m$-form  for a gauge theory is 
\begin{align*}
L(\s^*\omega, \s^*\vphi)=\tfrac{1}{2}\Tr[\s^*\Omega \w * (\s^*\Omega) ]+ \langle \s^*D\vphi, *\s^*D\vphi\rangle - U(\s^*\vphi),
\end{align*}
where $U$ is a potential term for the matter field, as is necessary for the spontaneous symmetry breaking (SSB) mechanism in the electroweak sector of the Standard Model. 

\subsection{Cartan geometry} 
\label{Cartan geometry} 

Connections $\omega$ on $\P$ such as described, known as Ehresmann or principal connections, are well suited to describe Yang-Mills fields of gauge theory. They are the heirs of another notion of connection, best suited  to describe gravity in a gauge theoretic way: Cartan connections. 
A Cartan connection $\varpi$ on a principal bundle $\P(M, H)$, beside satisfying the two defining properties of a principal connection,  defines an absolute parallelism on $\P$. A bundle equipped with a Cartan connection is a Cartan geometry, noted $(\P, \varpi)$. 

Explicitly, given a Lie algebra $\LieG \supset \LieH$ with dim $\LieG=$ dim $T_p\P$ for which a group is not necessarily chosen, a Cartan connection is $\varpi \in \Lambda^1(\P, \LieG)$ satisfying: $\varpi (X^v)=X$, $R^*_h\varpi =\Ad_{h\-}\varpi$ and $\varpi_p : T_p\P \rarrow \LieG$ is a linear isomorphism $\forall p \in \P$. This last defining property implies that the geometry of the bundle $\P$ is much more intimately related to the geometry of the base spacetime manifold $\M$, hence the fitness of Cartan geometry to describe gravity in the spirit of Einstein's insight. Concretely one can show that $T\M \simeq \P\times_H \LieG/\LieH$, and the image of $\varpi$ under the projection $\tau: \LieG \rarrow \LieG/\LieH$ defines a generalized soldering form, $\theta:=\tau(\varpi)$. The latter, more commonly known as the vielbein in the physics literature, implements (a version of) the equivalence principle and accounts for the specificities of gravity among other gauge interactions.  The $(\Ad, \LieG)$-tensorial curvature $2$-form $\Omega'$ of $\varpi$ is defined through the Cartan structure equation: $\Omega'=d\varpi + \tfrac{1}{2}[\varpi, \varpi]$. Its $\LieG/\LieH$-part is the torsion $2$-form $\Theta:=\tau(\Omega')$.

Given a $\Ad_H$-invariant bilinear form $\eta$ of signature $(r, s)$  on $\LieG/\LieH$, a $(r, s)$-metric $g$  on $\M$ is induced  by  $\varpi$ according to $g(X, Y):=\eta\left( \s^*\theta(X), \s^*\theta(Y)\right)$, for $X, Y \in T\M$ and $\s :\U\subset \M \rarrow \P$ a trivializing section. 

In the case $\LieG$ admits a $\Ad_H$-invariant splitting $\LieG=\LieH + \LieG/\LieH$, the Cartan geometry is said reductive. Then one has $\varpi=\omega+\theta$, where $\omega$ is a principal $H$-connection, and $\Omega'=\Omega + \Theta$ with $\Omega$ the curvature of $\omega$. 
As an example, the Cartan geometry with $(\LieG,\LieH)$ the Euclid and rotations Lie algebras is Riemann geometry with torsion.

Given a group $G$ and a closed subgroup $H$, $G/H$ is a homogeneous manifold and $G \xrightarrow{\pi} G/H$ is a $H$-principal bundle. The Maurer-Cartan form $\varpi_\text{\tiny G}$ on $G$ is a flat Cartan connection. So $(G, \varpi_\text{\tiny G})$ is a flat Cartan geometry, sometimes referred to as the Klein model for the geometry $(\P, \varpi)$,  which is thus said to be of type $(G, H)$. 

Let $V$ be a $(\LieG, H)$-module, i.e it supports a $\LieG$-action $\rho_*$ and a $H$-representation $\rho$ whose differential coincides with the restriction of the  $\LieG$-action  to $\LieH$. The Cartan connection defines a covariant derivative on $(\rho, V)$-tensorial forms. On sections of associated bundles, i.e on $\rho$-equivariant maps $\vphi$, we have: $D\vphi:=d\vphi + \rho_*(\varpi)\vphi$. As usual $D^2 \vphi=\rho_*(\Omega')\vphi$. On the curvature it gives the Bianchi identity: $D\Omega'=d\Omega'+[\varpi, \Omega']=0$.

The gauge group $\H \simeq \Aut_v(\P)$ acts on $\varpi$ and $\Omega'$ as it does on $\omega$ and $\Omega$ in \eqref{ActiveGT}.  The definition of local representatives via sections of $\P$, local gauge transformations and gluing properties thereof proceeds as in the standard case. 
\medskip


\subsection{The BRST framework}  
\label{The BRST framework}  

The infinitesimal version of \eqref{ActiveGT} can be captured by the so-called BRST differential algebra. 
 Abstractly \cite{Dubois-Violette1987} it is a   bigraded  differential algebra generated by  $\{\omega, \Omega, v, \chi\}$ where $v$ is the so-called ghost and the generators are respectively of degrees $(1, 0)$, $(2, 0)$, $(0, 1)$ and $(1, 1)$.  It is endowed with two nilpotent antiderivations $d$ and $s$, homogeneous of degrees $(1, 0)$ and $(0, 1)$ respectively, with vanishing anticommutator: $d^2=0=s^2$, $sd+ds=0$.  The algebra is equipped with a bigraded commutator $[\alpha, \beta]:=\alpha\beta -(-)^{\text{deg}[\alpha]\text{deg}[\beta]}\beta\alpha$. Notice that if the commutator vanishes identically, the BRST algebra is a  bigraded commutative differential  algebra.  
The action of $d$ is defined on the generators by: $d\omega=\Omega -\tfrac{1}{2}[\omega, \omega]$ (Cartan structure equation), $d\Omega=[\Omega, \omega]$ (Bianchi identity), $dv=\chi$ and $d\chi=0$. The action of the BRST operator on the generators gives the usual defining relations of the BRST algebra,
\begin{align}
\label{BRST}
s\omega=-dv -[\omega, v], \quad s\Omega=[\Omega, v], \quad  \text{ and } \quad sv=-\tfrac{1}{2}[v, v]. 
\end{align}
 Defining the degree $(1, 1)$ homogeneous antiderivation $\t d:=d+s$ and so-called algebraic connection $\t \omega:=\omega + v$,  \eqref{BRST} can be compactly rewritten as $\t \Omega:= \t d \t \omega + \tfrac{1}{2}[\t \omega, \t \omega]=\Omega$.  This is known as the ``russian formula''  \cite{Stora1984, Manes-Stora-Zumino1985} or ``horizontality condition'' \cite{Baulieu-TMieg1984, Baulieu-Bellon1986}. 
One is free to supplement this algebra with an element $\vphi$ of degrees $(0, 0)$  supporting a linear representation $\rho_*$ of the algebra as well as the action of the antiderivations,  so that upon defining $D:=d + \rho_*(\omega)$ one has consistently $D^2\vphi=\rho_*(\Omega)\vphi$ and 
\begin{align}
\label{BRST2}
s\vphi=-\rho_*(v)\vphi, \quad \text{ and } \quad sD\vphi=-\rho_*(v)D\vphi.
\end{align}

When the abstract BRST algebra is realized in the above differential geometric setup, the bigrading is according to the de Rham form degree and  ghost degree, $d$ is the de Rham differential on $\P$ (or $\M$) and $s$ is the de Rham operator on $\H$. The ghost is the Maurer-Cartan form on $\H$ so that $v \in \Lambda^1(\H, \text{Lie}\H)$, and given $\xi \in T\H$, $v(\xi) :\P \rarrow \LieH  \in \text{Lie}\H$ \cite{Bonora-Cotta-Ramusino}.  So in practice the ghost can be seen as a map $v:  \P\rarrow \LieH \in \text{Lie}\H$, a place holder that takes over the role of the infinitesimal gauge parameter. Thus the first two  relations of \eqref{BRST} and \eqref{BRST2} reproduce the infinitesimal gauge transformations of the gauge fields \eqref{ActiveGT}, while the third equation in \eqref{BRST} is the Maurer-Cartan structure equation for the gauge group $\H$. 

The BRST framework provides an algebraic way to characterize relevant quantities in gauge theories, such as admissible Lagrangian forms, observables and anomalies. Quantities of degree $(r, g)$ that are $s$-closed, that is $s$-cocycles $\in Z^{r,g}(s):=\ker s$, are gauge invariant. Quantities of degree $(r, g)$ that  are $s$-exact are $s$-coboundaries $\in B^{r,g}(s):=\text{Im } s$. Since $s^2=0$ obviously $B^{r,g}(s) \subset Z^{r,g}(s)$ and one defines the $s$-cohomology group $H^{r, g}(s):=Z^{r,g}(s)/B^{r,g}(s)$, elements of which  differing only by a coboundary, $c'=c+sb$, define the same cohomology class. Non-trivial Lagrangians and observables  must belong to $H^{n, *}(s)$.\footnote{If suitable boundary conditions are imposed on the fields of the theory or if the spacetime manifold is boundaryless,  the requirement of quasi-invariance of the Lagrangian, $sL=d\alpha$, is enough to ensure the invariance of the action, $\S=\int L$. So that one may consider $H^{r,g}(s|d)$, the $s$-modulo-$d$-cohomology instead of the strict $s$-cohomology.} For example, given a properly gauge invariant Yang-Mills Lagrangian $L$, $sL=0$, the prototypical Faddeev-Popov gauge-fixed Lagrangian is $L'=L+sb$, where $b$ is of degree $(n, -1)$ (since it involves an antighost, not treated here), and both belong to the same $s$-cohomology class in $H^{n, 0}(s)$. Wess-Zumino consistent gauge anomalies $\mathsf A$ - quantum gauge symmetry breaking of the quantum action $W=e^{iS}$, $sW=\mathsf A$ - belong to $H^{n,1}(s)$.

\section{Reduction of gauge symmetries: the dressing field method}   
\label{Reduction of gauge symmetries: the dressing field method}   

As insightful as the gauge principle is, gauge theories suffer from \emph{prima facie} problems such as an ill-defined quantization procedure due to the divergence of their path integral, and the masslessness of the interaction mediating fields (at odds with the phenomenology of the weak interaction). These drawbacks are rooted in the very thing that is the prime appeal of gauge theories: the gauge symmetry. Hence the necessity to come-up with strategies to reduce it. Broadly, two standard strategies to do so, addressing either problems respectively, are gauge fixings and SSB mechanisms. 
Furthermore, similarly to what happens in General Relativity, it may not be straightforward to extract physical observables in gauge theories. In GR, observables must be diffeomorphism-invariant. In gauge theories, observables must be gauge-invariant, e.g the abelian (Maxwell-Faraday) field strength or Wilson loops.

The dressing field approach is a third way, besides gauge fixing and SSB, to systematically reduce gauge symmetries. As such it may dispense to fix a gauge, can be a substitute to SSB (see \cite{Masson-Wallet,GaugeInvCompFields, Francois2014}) and provides candidate physical observables. 

\subsection{Composite fields}  
\label{Composite fields}  

Let $\P(\M, H)$ be a principal bundle equipped with a connection $\omega$ with curvature $\Omega$, and let $\vphi$ be a $\rho$-equivariant map on $\P$ to be considered as a section of the associated vector bundle $E=\P \times_H V$.  The gauge group is $\H \simeq \Aut_v(\P)$. The main content of the dressing field approach as a gauge symmetry reduction scheme is in the following

\begin{prop} 
\label{P1}
 If $K$ and $G$ are subgroups of $H$ such that $K\subseteq G \subset H$. Note $\K\subset \H$  the gauge subgroup associated with $K$. Suppose there exists a map
\begin{align} 
\label{DF}
u:\P \rarrow G \quad \text{ defined by its $K$-equivariance property }\quad  R_k^*u=k\-u,
\end{align}
  This map $u$, that we will call a \emph{dressing field}, allows to construct through $f: \P \rarrow \P$ given by $f(p)=pu(p)$, the following \emph{composite fields}
 \begin{align}
\label{CompFields}
\omega^u:&=f^*\omega=u\-\omega u+u\-du, \qquad \Omega^u:=f^*\Omega=u\-\Omega u =d\omega^u+\tfrac{1}{2}[\omega^u, \omega^u],\notag\\[1mm]
\vphi^u:&=f^*\vphi= \rho(u\-)\vphi \qquad \text{and} \qquad D^u\vphi^u:=f^*D\vphi=\rho(u\-)D\vphi=d\vphi^u+\rho_*(\omega^u)\vphi^u. 
\end{align}
 which are $\K$-invariant, $K$-horizontal and thus project on the quotient subbundle $\P/K \subset \P$.
 \end{prop}
 
  \noindent \textbf{NB}:  The dressing field can be \emph{equally defined} by its $\K$-gauge transformation: $u^\gamma=\Phi^*u=\gamma\-u$, with $\gamma \in \K\subset \H$.  
   This together with \eqref{ActiveGT} makes easy to check algebraically that the composite fields \eqref{CompFields} are $\K$-invariant indeed,  according to $(\chi^u)^\gamma=(\chi^\gamma)^{u^\gamma}=(\chi^\gamma)^{\gamma\-u}=\chi^u$.\footnote{We use $\chi=\{\omega, \Omega, \vphi, \ldots\} $ to denote a generic variable when performing an operation that applies equally well to any specific one.} 

Several comments are in order. First, in the event that $G\supset H$ then one has to assume that the $H$-bundle $\P$ is a subbundle of a $G$-bundle, and \emph{mutatis mutandis} the proposition still holds. Such a situation occurs when $\P$ is a reduction of a frame bundle  (of unspecified order) as the main object of this paper will illustrate. 

Second, if $K=H$ then the composited fields \eqref{CompFields} are $\H$-invariant, the gauge symmetry is fully reduced, and they live on $\P/H \simeq \M$. This shows that the existence of a global dressing field is a strong constraint on the topology of the bundle $\P$: a $K$-dressing field means that the bundle is trivial along the $K$-subgroup, $\P \simeq \P/K \times K$, while a $H$-dressing field means its triviality, $\P\simeq \M \times H$. 

Notice that despite the formal similarity with \eqref{ActiveGT} (or \eqref{PassiveGT}), the composite fields \eqref{CompFields} are not gauge transformed fields. Indeed the defining equivariance property \eqref{DF} of the dressing field implies $u\notin \H$, and $f\notin \Aut_v(\P)$. As a consequence, in general the composite fields do not belong to the gauge orbits of the original fields: $\chi^u \notin \mathcal{O}(\chi)$. The dressing field method then shouldn't be mistaken for a mere gauge fixing. 
\medskip

\subsection{Residual gauge symmetry}  
\label{Residual gauge symmetry}  

 Since in general $H/K$ is a coset, its action on the dressing field $u$ is left unspecified and depends on specifics of the situation at hand. Then in general nothing can be said of the transformation properties of the composite fields under $H/K$. But  interesting things happen if $K$ is a normal subgroup, $K\mathrel{\unlhd} H$, so that $H/K$ is a group that we note $J$ for convenience. The quotient bundle $\P/K$ is then a $J$-principal bundle noted $\P'=\P'(\M, J)$. 
We discuss two most important such cases in the following  subsections.
 
 \subsubsection{The composite fields as genuine gauge fields}  
 \label{The composite fields as genuine gauge fields}  

\begin{prop}  
\label{Prop2}
Let $u$ be a $K$-dressing field on $\P$. Suppose its $J$-equivariance  is given by 
\begin{align}
\label{CompCond}
R^*_ju=\Ad_{j\-}u, \qquad \text{ with } j \in J. 
\end{align}
Then the dressed connection $\omega^u$ is  a $J$-principal connection on $\P'$. That is, for $ X\in \LieJ$ and  $j\in J$, $\omega^u$ satisfies: $\omega^u(X^v)=X$ and $R^*_j\omega^u=\Ad_{j\-}\omega^u$. Its curvature is given by $\Omega^u$. 

\noindent Also, $\vphi^u$ is a $(\rho, V)$-tensorial map on $\P'$ and can be seen as a section of the associated bundle $E'=\P' \times_{J} V$. The covariant derivative on such sections is given by $D^u=d + \rho(\omega^u)$. 
\end{prop}
From this we immediately deduce the following

\begin{cor}
\label{Cor3}
The transformation of the composite fields under the residual $\J$-gauge symmetry is found in the usual way to be
\begin{align}
\label{GTCompFields}
(\omega^u)^{\gamma'}:&={\Phi'}^*\omega^u={\gamma'}\- \omega^u \gamma' + {\gamma'}\-d\gamma',  \qquad (\Omega^u)^{\gamma'}:={\Phi'}^*\Omega^u={\gamma'}\- \Omega^u
 \gamma', \notag\\
 (\vphi^u)^{\gamma'}:&={\Phi'}^*\vphi^u=\rho({\gamma'}\- )\vphi^u, \qquad \text{ and } \qquad  (D^u\vphi^u)^{\gamma'}:={\Phi'}^*D^u\vphi^u=\rho({\gamma'}\-) D^u\vphi^u,
\end{align}
with $\Phi' \in \Aut(\P') \simeq \J \ni \gamma'$. 
\end{cor}

\noindent \textbf{NB}: The relation $u^{\gamma'}={\gamma'}\- u \gamma'$ can be taken as an alternative to \eqref{CompCond} as a  condition on the dressing field $u$.

\paragraph{Further dressing operations} In the case where \eqref{CompCond} holds so that the composite fields \eqref{CompFields} are $\K$-invariant but genuine $\J$-gauge fields with residual gauge transformation given by \eqref{GTCompFields}, the question stands as to the possibility to perform a further dressing operation. 

Suppose a second dressing field $u'$ for the residual symmetry is available. It would be defined by ${u'}^{\gamma'}={\gamma'}\-u'$ for $\gamma'\in \J$. But in order to not spoil the $\K$-invariance obtained from the first dressing field $u$, the second dressing field should satisfy the \emph{compatibility condition} 
\begin{align}
\label{CompCond12}
R^*_ku'=u', \quad \text{ for } k\in K. \quad \text{ Or altenatively: } \quad {u'}^\gamma=u',\quad \text{ for } \gamma \in \K.
\end{align}
In this case indeed: 
\begin{align*}
\left(\chi^{uu'}\right)^\gamma&=\left(\chi^\gamma\right)^{u^\gamma {u'}^\gamma}=\left(\chi^\gamma\right)^{\gamma\-u u'}=\chi^{uu'}, \quad \gamma\in \K.\\
\left(\chi^{uu'}\right)^{\gamma'}&=\left(\chi^{\gamma'}\right)^{u^{\gamma'} {u'}^{\gamma'}}=\left(\chi^{\gamma'}\right)^{{\gamma'}\-u\gamma' \ {\gamma'}\-u'}=\chi^{uu'}, \quad \gamma' \in \J .
\end{align*}
We see that the defining properties of the dressing fields $u$ and $u'$, together with their compatibility conditions \eqref{CompCond} and \eqref{CompCond2} implies that $uu'$ can be treated as a single dressing for $\H$: 
\begin{align*}
\left(uu'\right)^{\gamma\gamma'}=\left(\left(uu'\right)^\gamma\right)^{\gamma'}=(\gamma\- uu')^{\gamma'}=\left(\gamma^{\gamma'}\right)\-\ {\gamma'}\-u\gamma'\ {\gamma'}\-u'={\gamma'}\-\gamma\- uu'= (\gamma\gamma')\- uu'.
\end{align*}
The extension of this scheme to any number of dressing field is straightforward, the details can be found in \cite{Francois2014}.

 \subsubsection{The composite fields as a new kind of gauge fields}  
 \label{The composite fields as  a new kind of gauge fields}  
 
 Before turning to this next case we need to introduce some definitions. Let $G' \supset G$ be a Lie group for which representations $(\rho, V)$ of $G$ are also representations of $G'$.
 Consider a $C^\infty$-map $C :P \times J \rarrow G'$, $(p, j) \mapsto C_p(j)$, satisfying
\begin{align}
\label{PropDefC}
C_p(jj')=C_p(j)C_{pj}(j').
\end{align} 
From this we have that $C_p(e)=e$,  $e$ the identity in both $J$ and $G'$, and $C_p(j)\-=C_{pj}(j\-)$. The differential of $C$ is 
\begin{align*}
dC_{(p, j)}=dC(j)_p + dC_{p| j}: T_p\P \oplus T_jJ \rarrow T_{C_p(j)}G',
\end{align*}
where $\ker dC(j)=T_jJ$ and $\ker dC_p=T_p\P$, with by definition
\begin{align*}
dC(j)_p(X_p)&:=\tfrac{d}{dt}C_{\phi_t}(j) |_{t=0}, \qquad \phi_t \text{ the flow of } X\in T\P \text{ and } \phi_{t=0}=p,\\
dC_{p|j}(Y_j)&:=\tfrac{d}{dt} C_p(\vphi_t)|_{t=0},\qquad \vphi_t \text{ the flow of } Y\in TJ \text{ and } \vphi_{t=0}=j.
\end{align*}
Notice that $C_p(j)\-dC_{(p,j)} : T_p\P \oplus T_jJ \rarrow  T_eG'=\LieG'$.
We are now ready to state our next result as the following
\begin{prop}
\label{P4}
Let $u$ be a $K$-dressing field on $\P$. Suppose its $J$-equivariance is given by
\begin{align}
\label{CompCond2}
(R_j^*u)(p)=j\-u(p)C_p(j),\quad \text{with } \quad j\in J \text{ and  $C$ a map as above}. 
\end{align}
Then $\omega^u$ satisfies
\begin{enumerate}
\item $\omega^u_p(X^v_p)=c_p(X):=\tfrac{d}{dt}C_p(e^{tX})|_{t=0}$, $\quad$ for $X\in \LieJ$ and $X^v_p \in V_p\P'$.
\item $R^*_j\omega^u=C(j)\- \omega^u C(j)+ C(j)\-dC(j)$.
\end{enumerate}
So, $\omega^u$ is  a kind of generalized connection $1$-form. Its curvature $\Omega^u$ is $J$-horizontal and satisfies $R^*_j\Omega^u=C(j)\-\Omega^uC(j)$. 
Also, $\vphi^u$ is a $\rho(C)$-equivariant map, $R^*_j \vphi^u=\rho\left(C(j)\right)\- \vphi^u$. The first order differential operator $D^u:=d + \rho_*(\omega^u)$ is a natural covariant derivative on such $\vphi^u$ so that $D^u\vphi^u$ is a $(\rho(C), V)$-tensorial form: 
$R^*_jD^u\vphi^u=\rho\left(C(j)\right)\- D^u\vphi^u$ and $(D^u\vphi^u)_p(X_p^v)=0$. 
\end{prop}

From this we can find the transformations of the composite fields under the residual gauge group $\J \simeq \Aut_v(\P')$. But first, we again need some preliminary results. Consider $\Phi' \in \Aut_v(\P')\simeq  \gamma' \in \J$, the residual gauge transformation of the dressing field is
\begin{align}
\label{GT-CDressField}
\left(u^{\gamma'} \right)(p):=(\Phi'^*u)(p)&=u(p\gamma'(p))=\gamma'(p)\- u(p) C_p\left(\gamma'(p)\right)=\left( {\gamma'}\- u C(\gamma') \right) (p).
\end{align}
\textbf{NB}: This relation can be taken as an alternative to \eqref{CompCond2} as a  condition on the dressing field $u$.\\
We witness the introduction of the map $C(\gamma'):\P' \rarrow G'$, $p \mapsto C_p\left(\gamma'(p)\right)$. It is given by the composition serie
\begin{align*}
&\P' \xrightarrow{\Delta} \P' \times \P' \xrightarrow{\id \times \gamma'} \P' \times J \xrightarrow{\phantom{b}C\phantom{b}} G', \\
& p\xmapsto{\phantom {bbbb}} (p, p) \xmapsto{\phantom{bbbb}}  \left(p, \gamma'(p)\right) \xmapsto{\phantom{bb}} C_p\left(\gamma'(p)\right).
\end{align*}
Its differential  $dC(\gamma') : T_p\P' \rarrow T_{C_p\left(\gamma'(p)\right)}G'$ is given by $dC(\gamma')=dC \circ \left( \id \oplus d\gamma' \right) \circ d\Delta$. 
Notice that we have $C_p\left(\gamma'(p) \right)\-dC(\gamma')_p : T_p\P' \rarrow T_eG'=\LieG'$. 
Then, we have the following

\begin{prop} 
\label{P5}
Given  $\Phi' \in \Aut_v(\P')\simeq  \gamma' \in \J$ and a dressing field $u$ satisfying \eqref{CompCond2} or \eqref{GT-CDressField}, the residual gauge transformations of the composite fields are
\begin{align}
\label{GTCompFields2}
(\omega^u)^{\gamma'}:&= \Phi'^*\omega^u = C(\gamma')\- \omega^u C(\gamma') + C(\gamma')\-dC(\gamma'), \qquad (\Omega^u)^{\gamma'}:=\Phi'^*\Omega^u= C(\gamma') \- \Omega^u C(\gamma'), \notag \\
(\vphi^u)^{\gamma'}:&= \Phi'^*\vphi^u = \rho\left( C(\gamma') \- \right) \vphi^u \qquad \text{ and } \qquad (D^u\vphi^u)^{\gamma'}= \Phi'^* D^u\vphi^u = \rho\left( C(\gamma') \- \right) D^u\vphi^u.
\end{align}
So, the composite fields \eqref{CompFields} behave as gauge fields of a new kind, and implement the \emph{gauge principle} - or principle of \emph{local symmetry} - of field theory in Physics. 
\end{prop}

\noindent \textbf{NB}: Under a further gauge transformation $\Psi \in \Aut_v(\P') \simeq \eta\in \J$, the dressing field behaves as
\begin{align*}
\left(\Psi^*(\Phi^*u)\right)(p)&= \left( (\Phi \circ \Psi)^*u \right)(p)=u\left(\Phi(p\eta(p)\right)=u\left(\Phi(p)\eta(p)\right)=u\left(p\gamma(p)\eta(p)\right)= \eta(p)\-\gamma(p)\- u(p)C_p\left(\gamma(p)\eta(p)\right),\\
							&=\left(\eta\-\gamma\-\ u\ C\left(\gamma\eta \right)\right)(p).\\
\text{or }  \left(\Psi^*(\Phi^*u)\right)(p)&= \left( \gamma\- u C(\gamma) \right) (\Psi(p))=\gamma\left(p\eta(p) \right)\- u\left(p\eta(p)\right) C_{p\eta(p)}\left(\gamma(p\eta(p) \right),\\
							&= \eta(p)\-\gamma(p)\-\eta(p) \cdot \eta(p)\-u(p) C_p\left( \eta(p)\right) \cdot C_{p\eta(p)}\left(\eta(p)\- \gamma(p) \eta(p) \right),\\
							&=\eta(p)\-\gamma(p)\-\ u(p) C_p\left( \eta(p)\right) \cdot C_{p\eta(p)}\left(\eta(p)\- \right) C_p\left(\gamma(p) \eta(p) \right) = \eta(p)\-\gamma(p)\- u(p)C_p\left(\gamma(p)\eta(p)\right).
\end{align*}
This secures the fact that the action  \eqref{GTCompFields2} of the residual gauge symmetry on the composites fields is well-behaved as a representation.

\paragraph{The case of  $1$-$\alpha$-cocycles} Suppose $C_p : J \rarrow G'$ is defined by $C_p(jj')=C_p(j)\ \alpha_j[C_p(j')]$, for $\alpha: J \rarrow \Aut(G')$ a continuous group morphism. Such objects appear in the representation theory of crossed products of $C^*$-algebras and is known as a \emph{$1$-$\alpha$-cocycle} (see \cite{Pedersen79,Williams07}).\footnote{In the general theory the group $G'$ is replaced by a $C^*$-algebra $A$.} Its defining property is an example of \eqref{PropDefC}, and everything that has been said in this section -and will be said in the following - applies when $C_p$ is a $1$-$\alpha$-cocycle.
As a particular case, consider the following

\begin{prop}
\label{alpha_cocycle}
 Suppose $J$ is abelian and let $A_p, B : J\rarrow GL_n$  be group morphisms where $R^*_jA_p(j')=B(j)\- A_p(j') B(j)$. Then $C_p:=A_pB : J\rarrow GL_n$ is a $1$-$\alpha$-cocyle where the morphism $\alpha: J \rarrow \Aut(GL_n)$ is the conjugate action through the morphism $B$: $\alpha_j [g] =B(j)\-[g] B(j)$, with $g \in GL_n$.
 \end{prop}
As a matter of fact,  in the case of the conformal Cartan geometry and the associated tractors and twistors, $1$-$\alpha$-cocycles of this type do appear, with $J$ is the Weyl group of rescalings.

\subsection{Application to the BRST framework} 
\label{Application to the BRST framework} 

The BRST algebra encodes the infinitesimal gauge symmetry. It is to be expected that the dressing field method modifies it. To see how, let us first consider the following 
\begin{prop}    
Given the BRST algebra \eqref{BRST}-\eqref{BRST2} on the initial gauge variables and the ghost $v\in Lie\H$. The composite fields \eqref{CompFields} satisfy the  \emph{modified BRST algebra}:
\begin{align}
\label{NewBRST}
&s\omega^u=-D^uv^u=-dv^u-[\omega^u, v^u], \quad s\Omega^u=[\Omega^u, v^u], \quad s\vphi^u=-\rho_*(v^u)\vphi^u,\quad \text{and } \quad sv^u=-\tfrac{1}{2}[v^u, v^u]\\[2mm]
 &\text{with the \emph{dressed ghost} } \quad v^u=u\-vu + u\-su.  \notag
\end{align}
This result does not rest on the assumption that $u$ is a dressing field.
\end{prop}

\begin{proof} 
The result is easily found by expressing the initial gauge variable $\chi=\{ \omega, \Omega, \vphi\}$ in terms of the dressed fields $\chi^u$ and the dressing field $u$, and re-injecting in the initial BRST algebra  \eqref{BRST}-\eqref{BRST2}. At no point of the derivation does $su$ need to be explicitly known. It then holds regardless if $u$ is a dressing field or not. 
\end{proof}

If the ghost $v$ encodes the infinitesimal initial $\H$-gauge symmetry, the dressed ghost $v^u$ encodes the infinitesimal residual gauge symmetry. Its concrete expression depends on the BRST transformation of $u$. 

Under the hypothesis $K\subset H$, the ghost decomposes as $v=v_\LieK + v_{\LieH/\LieK}$, and the BRST operator splits accordingly: $s=s_\LieK + s_{\LieH/\LieK}$. If $u$ is a dressing field its BRST transformation is the infinitesimal version of its defining transformation property: $s_\LieK u=-v_\LieK u$. So the dressed ghost is
\begin{align*}
v^u=u\-vu+u\-su=u\-(v_\LieK + v_{\LieH/\LieK})u+u\-(-v_\LieK u+ s_{\LieH/\LieK}u)=u\- v_{\LieH/\LieK} u+u\- s_{\LieH/\LieK}u. 
\end{align*}
We see that the Lie$\K$ part of the ghost, $v_\LieK$, has disappeared. This means that $s_\LieK \chi^u=0$, which expresses the $\K$-invariance of the composite fields \eqref{CompFields}. 

\paragraph{Residual BRST symmetry}  

In general $\LieH/\LieK$ is simply a vector space, so $s_{\LieH/\LieK} u$ is left unspecified and nothing can be said in general of $v^u$ and of the form of the modified BRST algebra \eqref{NewBRST}. But following  section \ref{Residual gauge symmetry}, if $\K\mathrel{\unlhd}H$ then $H/K=J$ is a group with lie algebra $\LieH/\LieK=\LieJ$. We here provide the BRST treatment of the two cases detailed in this section.

Suppose the dressing field satisfies the condition \eqref{CompCond}, whose BRST version is: $s_{\LieJ} u=[u, v_{\LieJ}]$. The dressed ghost is then
\begin{align}
\label{NewGhost2}
v^u=u\- v_{\LieJ} u+u\- s_{\LieJ}u=u\- v_{\LieJ} u+u\- (u v_{\LieJ} -  v_{\LieJ}u)=v_{\LieJ}. 
\end{align}
This in turn implies that the new BRST algebra is
\begin{align}
\label{NewBRST2}
s\omega^u=-D^uv_{\LieJ}=-dv_{\LieJ}-[\omega^u, v_{\LieJ}], \quad s\Omega^u=[\Omega^u, v_{\LieJ}], \quad s\vphi^u=-\rho_*(v_{\LieJ})\vphi^u,\quad \text{and } \quad sv_{\LieJ}=-\tfrac{1}{2}[v_{\LieJ}, v_{\LieJ}].
\end{align}
This is the BRST version of \eqref{GTCompFields}, and reflects the fact that the composites fields \eqref{CompFields} are genuine $\J$-gauge fields, in particular that $\omega^u$ is a $J$-connection. 

A further dressing field $u'$ would be defined by $s_{\LieJ} u'=-v_{\LieJ} u'$, and the necessary compatibility condition it needs to satisfy is $s_\LieK u'=0$. The combined dressing $uu'$ is such that $suu'=-vuu'$, so that $v^u=0$ and $s\chi^{uu'}=0$. Again the straightforward extension of the scheme to any number of dressing fields can be found in \cite{Francois2014}. 
\medskip

Suppose now that the dressing field satisfies the  condition \eqref{CompCond2}, whose BRST version is: $s_{\LieJ} u=-v_{\LieJ}u + uc_p(v_{\LieJ})$. The dressed ghost is then
\begin{align}
\label{NewGhost3}
v^u=u\- v_{\LieJ} u+u\- s_{\LieJ}u=u\- v_{\LieJ} u+u\- \left(-v_{\LieJ}u + uc_p(v_{\LieJ})\right)=c_p(v_{\LieJ}). 
\end{align}
This in turn implies that the new BRST algebra is
\begin{align}
\label{NewBRST3}
&s\omega^u=-dc_p(v_{\LieJ})-[\omega^u, c_p(v_{\LieJ})], \quad s\Omega^u=[\Omega^u, c_p(v_{\LieJ})], \quad s\vphi^u=-\rho_*(c_p(v_{\LieJ}))\vphi^u,\\[1mm]
\quad \text{and } \quad &sc_p(v_{\LieJ})=-\tfrac{1}{2}[c_p(v_{\LieJ}), c_p(v_{\LieJ})]. \notag
\end{align}
This is the BRST version of \eqref{GTCompFields2}, and reflects the fact that the composites fields \eqref{CompFields} instantiate the gauge principle in a satisfactory way.
%

\subsection{Local aspects and Physics} 
\label{Local aspects and Physics} 

 Until now we have exposed  the global aspects of the dressing approach on the bundle $\P$  to emphasize the geometric nature of the composites fields obtained, according to the given equivariance properties displayed by the dressing field. Most notably we stated that  the composite field can behave as a new kind of gauge fields. 
 
But to do Physics we need the local representatives on an open subset $\U\subset \M$ of global dressing and composite fields. These are obtained in the usual way from a local section $\sigma:\U \rarrow \P$ of the bundle. The  important properties they thus retain is their gauge invariance and residual gauge transformations. 

If it happens that a dressing field is defined locally on $\U$ first, and not directly on $\P$, then the local composite fields $\chi^u$ are defined in terms of the local dressing field $u$ and local gauge fields $\chi$ by \eqref{CompFields}. The gauge invariance and residual gauge transformations of these local composite fields are derived from the gauge transformations of the local dressing field under the various subgroups of the local gauge group $\H_\text{\tiny{loc}}$ according to $(\chi^u)^\gamma=(\chi^\gamma)^{u^\gamma}$. The BRST treatment for the local objects mirrors exactly the one given for the global objects.  

This being said, note $A=\sigma^*\omega$, $F=\sigma^*\Omega$ for definiteness but keep $u$ and $\vphi$ to denote the local dressing field and section. We state the final proposition of this section, dealing with gauge theory.
\begin{prop} 
\label{Prop-Lagrangian}
Given the geometry defined by a bundle $\P(\M, H)$ endowed with $\omega$ and the associated bundle $E$, suppose we have a gauge theory given by the prototypical $\H_\text{\tiny{loc}}$-invariant Yang-Mills Lagrangian
\begin{align*}
L(A, \vphi)=\tfrac{1}{2}\Tr(F \w * F) + \langle D\vphi, \ *D\vphi\rangle - U(||\vphi||), 
\end{align*}
where $ ||\vphi||:=|\langle \vphi \rangle |^{\sfrac{1}{2}}$. If there is a local dressing field $u: \U \rarrow G  \subset H$ with $\K_\text{\tiny{loc}}$-gauge transformation $u^\gamma=\gamma\-u$, then the above Lagrangian is actually  a $\H_\text{\tiny{loc}}/\K_\text{\tiny{loc}}$-gauge theory defined in terms of  $\K_\text{\tiny{loc}}$-invariant variables since we have
\begin{align*}
L(A, \vphi)=L(A^u, \vphi^u)=\tfrac{1}{2}\Tr(F^u \w * F^u) + \langle D^u\vphi^u, \ *D^u\vphi^u\rangle - U(||\vphi^u||)
\end{align*}
by a mere change of variables. 
\end{prop}
\begin{proof}
The result follows straightforwardly from the $\H_\text{\tiny{loc}}$-invariance of the initial Lagrangian. Since $L(A^\gamma, \vphi^\gamma)=L(A, \vphi)$ for $\gamma:\U \rarrow H$,
 holds as a formal property of $L$, it follows that $L(A^u, \vphi^u)=L(A, \vphi)$ for $u:\U \rarrow G \subset H$.
\end{proof}

Notice  that since $u$ is a dressing field, $u \notin \H_\text{\tiny{loc}}$ so the dressed Lagrangian $L(A^u, \vphi^u)$ ought not to be confused with a gauge-fixed Lagrangian $L(A^\gamma, \vphi^\gamma)$ for some chosen $\gamma \in \H_\text{\tiny{loc}}$, even if it may happen that $\gamma=u$.
 A fact that might go unnoticed. As we've stressed in the opening of section \ref{Reduction of gauge symmetries: the dressing field method}, the dressing field approach is distinct from both gauge-fixing and spontaneous symmetry breaking as a means to reduce gauge symmetries. 

 Let us highlight the fact that a dressing field can often be constructed  by requiring the gauge invariance of a prescribed ``gauge-like condition''. 
  Such a condition is given when a local gauge field $\chi$ (often the gauge potential)  transformed by a field $u$ with value in the symmetry group $H$, or one of its subgroups, is required to satisfy a functional constraint: $\Sigma(\chi^u)=0$.  Explicitly solved, this makes $u$ a function of $\chi$, $u(\chi)$, thus sometimes called \emph{field dependent gauge transformation}. However this terminology is valid if and only if  $u(\chi)$ transforms under the action of $\gamma\in \H_\text{\tiny{loc}}$ as $u(\chi)^\gamma:=u(\chi^\gamma)=\gamma\- u(\chi) \gamma$, in which case $u(\chi) \in \H_\text{\tiny{loc}}$.  But if the functional constraint  still holds under the action of $\H_\text{\tiny{loc}}$, or of a subgoup thereof, it follows  that $(\chi^\gamma)^{u^\gamma}=\chi^u$ (or equivalently that $s\chi^u=0$). This in turn imposes that $u^\gamma=\gamma\- u$ (or $su=-vu$) so that $u \notin \H_\text{\tiny{loc}}$ but is indeed a dressing field.

 This and the above proposition generalizes the pioneering idea of Dirac \cite{Dirac55, Dirac58}  aiming at quantizing QED by rewriting the classical theory in terms of gauge-invariant variables. The idea was rediscovered several times, early by Higgs himself \cite{Higgs66} and Kibble \cite{Kibble67}. The invariant variables were sometimes termed \emph{Dirac variables} \cite{Pervushin, Lantsman} and reappeared in various contexts in gauge theory, such as  QED \cite{Lavelle-McMullan93},  quarks theory in QCD \cite{McMullan-Lavelle97}, the proton spin decomposition controversy  \cite{LorceGeomApproach, Leader-Lorce, FLM2015_I} and most notably  in electroweak theory and Higgs mechanism \cite{Frohlich-Morchio-Strocchi81, McMullan-Lavelle95, Chernodub2008, Faddeev2009, Masson-Wallet, Ilderton-Lavelle-McMullan2010, Struyve2011, vanDam2011}. Indeed,  proposition \ref{Prop-Lagrangian} applies to the electroweak sector of the Standard Model and thus provides an alternative  to the usual textbook interpretation of the Higgs mechanism in terms of spontaneous symmetry breaking, see \cite{GaugeInvCompFields, Francois2014}. 
  
The dressing field approach thus gives a unifying and clarifying framework for these works, and others concerning the BRST treatment of anomalies in QFT \cite{Manes-Stora-Zumino1985, Garajeu-Grimm-Lazzarini}, Polyakov's ``partial gauge fixing'' for $2D$-quantum gravity \cite{Polyakov1989, Lazzarini2008} or the construction of the Wezz-Zumino functionnal \cite{Attard-Lazz2016}.
 It is the aim of this paper and  its companion to show that both tractors and twistors can also be encompassed by this approach, which furthermore highlights their nature as gauge fields of a non-standard kind. The case of twistors is dealt with in the next section.

\section{Twistors from conformal Cartan geometry via dressing} 
\label{Twistors from conformal Cartan geometry via dressing} 

 Due to the important progress of the last twenty years, the term \emph{twistor} is now more general than it used to.
 Twistor theory for various differential geometric structures has been devised, e.g for paraconformal manifolds in \cite{Bailey-Eastwood91}.  
 And the reference text \cite{Cap-Slovak09} defines it for any parabolic geometry.

  However, as mentioned in our introduction, initially the twistor bundle was devised for conformal manifolds and constructed via prolongation of a defining differential equation. A procedure deemed more  explicit for calculational purposes than the  bundle construction \cite{Bailey-Eastwood91}. 
  We briefly review this procedure in the following subsection, so that the reader can compare with the derivation via the dressing field method in the next.

\subsection{Bottom-up construction via prolongation of the Twistor Equation: a reminder} 
\label{Bottom-up construction via prolongation of the Twistor Equation: a reminder} 

We  essentially follow the expositions and use the notations of \cite{Penrose-Rindler-vol2} section $6.9$, and \cite{Bailey-Eastwood91} section $6.1$. But first we need to remind how points of Minkowski space $\sM:=(\RR^4, \eta)$ are represented as hermitian matrices and how the actions of the Lorentz group and Lie algebra are represented. This will  be useful latter on. 

Let $x=x^a$ the column vector representing, in abstract index notation, the coordinates of a point in $\sM$ w.r.t any basis  $\{e_a\}_{a=0, \ldots, 3}$.
The spacetime interval is given by $\lVert x \rVert^2=(x^0)^2 - (x^1)^2 - (x^2)^2 -(x^3)^2 =x^T\eta x$. An element $S={S^a}_b$ of the Lorentz group $S\!O(1, 3)=\left\{S \in GL_4(\RR)\ | S^T\eta S=\eta  \right\}$ acts as: $x'=Sx \rarrow {x'}^a={S^a}_b x^b$. Its Lie algebra $\so(1, 3)=\left\{ s\in GL_4(\RR)\ | \  s^T\eta + \eta s=0  \right\}$ acts likewise: $x'=sx \rarrow {x'}^a={s^a}_b x^b$.

Consider $\left\{\s_a^{AA'}\right\}_{a=0, \ldots, 3}$, a basis of $2\times2$ hermitian matrices Herm$(2, \CC)=\left\{ M \in GL_2(\CC)\ | \ M^*=M  \right\}$, where $*$ is the operation of transposition-conjugation. As vector spaces, $\sM$ and Herm$(2, \CC)$ are isomorphic via 
\begin{align}
\label{MorphMinkowski}
\sM &\rarrow \text{Herm}(2, \CC),  \notag\\
x=x^a &\mapsto \b x=\b x^{AA'}:=x^a \s_a^{AA'} = \tfrac{1}{2}\begin{pmatrix} x^0+x^3 & x^1 - ix^2 \\ x^1 + ix^2 & x^0 - x^3 \end{pmatrix}. 
\end{align}
Upper case Latin letters are Weyl spinor indices, taking values $0$ and $1$. The spacetime interval is then given by $\lVert x\rVert^2=4 \det(\b x)$. The action of $S\!O(1, 3)$ on $x$ preserving $\eta$  is represented by the action of $S\!L(2, \CC)$ on $\b x$ preserving~$\det$:
\smallskip

\begin{align}
\label{GroupMorph-Lorentz-SL}
\begin{split}
S\!O(1, 3) \times \sM  &\rarrow \sM, \\
(S, x) &\mapsto Sx  
\end{split}
\qquad \qquad \qquad \Rightarrow \hspace{-5mm}
\begin{split}
S\!L(2, \CC) \times \text{Herm}(2, \CC) &\rarrow \text{Herm}(2, \CC),  \\
(\b S, \b x) &\mapsto \b S \b x \b S^*
\end{split}
\end{align}
Since $\b S$ and $-\b S$ represent the same Lorentz transformation, the homomorphism $S\!O(1, 3) \rarrow S\!L(2, \CC)$ is $1:2$ covering. It is a spin representation of the Lorentz group. 
The action of $\so(1, 3)$ is likewise represented by the action of $\sl(2, \CC)$:
\begin{align}
\label{LieMorph-Lorentz-SL}
\begin{split}
\so(1, 3) \times \sM  &\rarrow \sM, \\
(s, x) &\mapsto sx  
\end{split}
\qquad \qquad \qquad \Rightarrow \hspace{-5mm}
\begin{split}
\sl(2, \CC) \times \text{Herm}(2, \CC) &\rarrow \text{Herm}(2, \CC),  \\
(\b s, \b x) &\mapsto \b s \b x + \b x \b s^*
\end{split}
\end{align}

Now, to construct twistors, one starts with a $4$-dimensional conformal manifold $(\M, c)$ with $c$ the conformal class of the Levi-Civita connection. 
Tensorial indices, represented by lower case Greek indices $(\mu,\nu, \ldots)$, can be converted to Minkowski indices $(a, b, \ldots )$ via a tetrad/vierbein field ${e^a}_\mu$ (related to a choice of metric $g\in c$). Then Minkowski indices can then be converted into spinor indices $(AA', BB', \ldots )$ via the isomorphism  \eqref{MorphMinkowski}.
One then defines the Twistor Equation as
\begin{align}
\label{TE}
{\nabla^{(A}}_{A'}\omega^{B)}=0, \quad \text{ or equivalently as}\quad  \nabla_{AA'}\omega^B  -  \tfrac{1}{2}\delta^B_A\ \nabla_{CA'}\omega^C=0
\end{align}
 where  $\nabla$ is the Levi-Civita connection associated to a choice of metric $g \in c$ and $\omega^B: \M \rarrow \CC^2$ is a Weyl spinor. This is the differential equation to be prolonged and recast as a system of differential equations.
To do so one defines the intermediary dual spinor variable $\pi_{A'}:= \tfrac{i}{2}\nabla_{CA'}\omega^C$ so that the Twistor Equation is
\begin{align}
\label{TE2}
\nabla_{AA'}\omega^B + i\delta^B_A \pi_{A'}=0.
\end{align}
One only has to find a constraint equation on $\pi_{A'}$ to close the system. This is done by applying $\nabla$ again, and after some algebra equation \eqref{TE} is replaced by the linear system
\begin{align}
\label{TE3}
\nabla_{AA'}\omega^B + i\delta^B_A \pi_{A'}=0, \qquad  \nabla_{AA'}\pi_{B'}- i \b P_{AA'BB'}\omega^B =0,
\end{align}
where $\b P_{AA'BB'} \simeq P_{ab}:=-\tfrac{1}{2}\left( R_{ab} - \tfrac{1}{6}Rg_{ab} \right)$ is the Schouten tensor.  
This system can be rewritten as the action of a linear operator $\nabla_{AA'}^\sT$ on the bi-spinor $Z^\alpha=(\omega^B, \pi_{A'})\in \CC^4$
\begin{align}
\label{TE4}
\nabla^\sT_{AA'}Z^\alpha =0, \qquad \Rightarrow \qquad \nabla_{AA'} \begin{pmatrix}[1.2] \omega^B \\ \pi_{B'} \end{pmatrix}   +  \begin{pmatrix}[1.2] 0  & i\delta_A^B \\ -i \b P_{AA'BB'} & 0 \end{pmatrix}  \begin{pmatrix}[1.2]\omega^B  \\ \pi_{A'}  \end{pmatrix} =0.
\end{align}
Given the Weyl rescaling of the metric $\h g=z^2 g$, after some algebra one finds that the connection changes as $\h \nabla_{AA'} X^C= \nabla_{AA'} X^C + \delta^C_A\b\Upsilon_{DA'}X^D $ and the Schouten tensor changes as ${\h {\b P}}_{AA'BB'}={\b P}_{AA'BB'} + \nabla_{AA'} \b \Upsilon_{BB'}- \b \Upsilon_{AB'} \b \Upsilon_{BA'}$, with $\b \Upsilon_{AA'}:=z\- \d_{AA'} z$  ($\Upsilon_a:=\d_a \ln z$).
So by \emph{requiring} the conformal invariance of the first spinor component, one fixes the conformal transformation of the second spinor component:
\begin{align}
\h \omega^A&= \omega^A, \notag\\[-7mm]
& \qquad \qquad \qquad \qquad \qquad \qquad   \text{ Or in matrix form, }\qquad  \begin{pmatrix}[1.2]  \h \omega^A \\ \h \pi_{A'} \end{pmatrix} =  \begin{pmatrix}[1.2] \1 & 0 \\ i\b\Upsilon_{AA'} & \1\end{pmatrix}\begin{pmatrix}[1.2]  \omega^A \\ \pi_{A'} \end{pmatrix}. \label{GTtwistor}\\[-7mm]
\h\pi_{A'}&=\pi_{A'} + i\b\Upsilon_{AA'}\omega^A. \notag
\end{align} 
 This, one may consider as a gauge transformation so that \emph{generic} bi-spinors $Z^\alpha=(\omega^A,  \pi_{A'})$ gauge-related by \eqref{GTtwistor}, called \emph{twistors}, are considered as sections  of a vector bundle over $(\M, c)$ with fiber $\CC^{4}$: the  \emph{local twistor bundle} $\sT$. 

With still more algebra and the relation $\h\nabla_{AA'}X_{C'}= \nabla_{AA'}X_{C'}- \Upsilon_{C'A} X_{A'}$, one finds that this gauge equivalence holds also for the bi-spinor defined by \eqref{TE3}
\begin{align}
\label{GTtwistor-connection}
\begin{pmatrix}[1.2]  \h{(\nabla_{AA'}\omega^D + i\delta^D_A \pi_{A'})} \\ \h{( \nabla_{AA'}\pi_{B'}- i \b P_{AA'BB'}\omega^B)} \end{pmatrix}  =  \begin{pmatrix}[1.2] \1 & 0 \\ i\b\Upsilon_{DB'} & \1\end{pmatrix}\begin{pmatrix}[1.2] \nabla_{AA'}\omega^D + i\delta^D_A \pi_{A'} \\ \nabla_{AA'}\pi_{B'}- i\b P_{AA'BB'}\omega^B \end{pmatrix}.
\end{align}
So the linear operator $\nabla_{AA'}^\sT$  \eqref{TE4} defines a covariant derivative on $\sT$  called the twistor transport or \emph{twistor connection}. A twistor satisfying $\nabla_{AA'}^\sT Z^\alpha=0$ is a global twistor and coincides with the notion of twistor from which Penrose wanted  to derive Minkowski space.

There is furthemore a well defined bilinear form on twistors $Z, Z' \in \Gamma(\sT)$ defined by
\begin{align}
\label{metric_twistor}
\langle Z, Z'  \rangle:= \pi_A^* {\omega'}^A + \omega^{A'*} \pi'_{A'} .
\end{align}
 Indeed it is invariant under Weyl rescaling $\langle \h Z, \h {Z'}\rangle=\langle Z, Z' \rangle$, as can be verified via \eqref{GTtwistor}. One also checks via \eqref{TE4} that, like a Levi-Civita connection, the twistor connection preserves the metric thus defined since $\nabla_\sT \langle Z, Z' \rangle = 2 \langle \nabla^\sT Z, Z' \rangle$. In  physics, the  \emph{helicity} of a twistor $Z$, usually noted $s$, is defined as half its norm: $\langle Z, Z\rangle$=2s. 

The commutator of the twistor connection defines the \emph{local twistor curvature}  $\mathsf{K}$,
\begin{align}
\label{twistor_curvature}
\left[\nabla^\sT, \nabla^\sT \right]Z=  \mathsf{K} 
 Z=\begin{pmatrix}[1.2] \ \ \b W & 0 \\ -i \b C & -\b W^*\end{pmatrix}  \begin{pmatrix}[1.2] \omega \\ \pi \end{pmatrix},
\end{align}
where $C=\nabla P$ is the Cotton tensor, and $W$ is the Weyl tensor. 
From this one sees immediately that the twistor connection $\nabla_\mu^\sT$ is flat if and only if $(\M, c)$ is conformally flat. 
\medskip 

This is how is constructed the local twistor bundle $\sT$ endowed with the twistor connection $\nabla^\sT$, bottom up from the Twistor  Equation on a conformal manifold $(\M, c)$.  This approach, while presenting the advantage of being explicit, involves some amount of computation in order to derive the basic objects and their transformation properties. In the next section we lay our case that  objects very much like these can be recovered with much less computation, top-down from a gauge structure over $\M$ via the dressing field method. By doing so, the nature of the twistors and twistor connection as gauge fields of the non-standard kind described in section \ref{The composite fields as  a new kind of gauge fields}  is made clear.

\subsection{Top-down gauge theoretic approach via dressing} 
\label{Top-down gauge theoretic approach via dressing} 

The gauge structure on spacetime $\M$ that we start with is the conformal Cartan geometry and its associated spin bundle.   We describe it in the following subsection, and the dressing field method is applied in the next.

\subsubsection{The conformal Cartan bundle and its spin vector bundle}  
\label{The conformal Cartan bundle and its spin vector bundle}  

Since we are ultimately interested in twistors, we are concerned with $4$-dimensional base manifolds $\M$ despite the fact that the conformal Cartan bundle can be defined for dimension $\geq 3$. The conformal Cartan geometry $(\P, \varpi)$ is said modeled on the Klein model $(G, H)$ where $G=PS\!O(2, 4)=\left\{ M \in GL_{6}(\RR)\  | \ M^T \Sigma M= \Sigma, \det{M}=1 \right\}/ \pm \id$ with $\Sigma=\begin{psmallmatrix}  0 & 0 & -1 \\ 0 & \eta & 0 \\ -1 & 0 & 0 \end{psmallmatrix}$, $\eta$ the flat metric of signature $(1, 3)$, and $H$ is a parabolic subgroup such that the Homogeneous space $G/H \simeq (S^1 \times S^3) / \mathbb{Z}^2$ is the conformal compactification of Minkowski space, $\sM$. The structure group of the conformal Cartan bundle $\P(\M, H)$ comprises Lorentz, Weyl and conformal boost symmetries and has the following matrix presentation \cite{Cap-Slovak09, Sharpe}
\begin{align*}
 H = K_0\, K_1=\left\{ \begin{pmatrix} z &  0 & 0  \\  0  & S & 0 \\ 0 & 0 & z^{-1}  \end{pmatrix}\!  \begin{pmatrix} 1 & r & \tfrac{1}{2}rr^t \\ 0 & \1 & r^t \\  0 & 0 & 1\end{pmatrix}  \bigg|\ z\in W:=\RR^*_+,\ S\in S\!O(1, 3), 
\ r\in \RR^{4*} \right\}.
\end{align*} 
Here ${}^t$ stands for the $\eta$-transposition, namely for the row vector $r$ one has $r^t = (r \eta^{-1})^T$ (the operation ${}^T\,$ being the usual matrix transposition), and $\RR^{4*}$ is the dual of $\RR^4$.  
Clearly $K_0\simeq CO(1, 3)$ via $(S, z) \rarrow zS$, and $K_1$ is the abelian group of conformal boosts. 
The corresponding Lie algebras  $(\LieG, \LieH)$ are graded \cite{Kobayashi}:  $[\LieG_i, \LieG_j] \subseteq \LieG_{i+j}$, $i,j=0,\pm 1$ with the abelian Lie subalgebras $[\LieG_{-1}, \LieG_{-1}] = 0 = [\LieG_1, \LieG_1]$. They decompose respectively as, $\LieG=\LieG_{-1}\oplus\LieG_0\oplus\LieG_1 \simeq \RR^4\oplus\co(1, 3)\oplus\RR^{4*}$ and $\LieH=\LieG_0\oplus\LieG_1 \simeq \co(1, 3)\oplus\RR^{4*}$. In matrix notation we have,
\begin{align*}
\mathfrak{g} = \left\{ 
\begin{pmatrix} \epsilon &  \rho & 0  \\  \tau  & s & \rho^t \\ 0 & \tau^t & -\epsilon  \end{pmatrix} \bigg|\ (s-\epsilon\1)\in \mathfrak{co}(1, 3),\ \tau\in\mathbb{R}^4,\ \rho\in\mathbb{R}^{4*}  
\right\} 
\supset
\LieH = \left\{ \begin{pmatrix} \epsilon &  \rho & 0  \\  0  & s & \rho^t \\ 0 & 0 & -\epsilon  \end{pmatrix} \right\},
\end{align*} 
with the $\eta$-transposition $\tau^t = (\eta\tau)^T$ of the  column vector $\tau$.
The graded structure of the Lie algebras is automatically handled by the matrix commutator.

The Cartan bundle $\P$ is then endowed with the conformal Cartan connection, whose local representative on $\U \subset \M$ is $\varpi  \in \Lambda^1(\U , \LieG)$ with curvature $\Omega\in \Lambda^2(\U, \LieG)$. In matrix representation
\begin{align*}
\varpi =\begin{pmatrix} a & P & 0 \\ \theta & A & P^t \\0 & \theta^t & -a \end{pmatrix}, \qquad \text{and} \quad  \Omega=d\varpi+\varpi^2=\begin{pmatrix} f & C & 0 \\ \Theta & W & C^t \\0 & \Theta^t & -f \end{pmatrix}.
\end{align*}
The soldering part of $\varpi$ is $\theta=e\cdot dx$, i.e  $\theta^a:={e^a}_\mu dx^\mu$, with $e={e^a}_\mu$ the so-called vierbein or tetrad field.\footnote{ Notice that from now on we shall make use of  ``$\cdot$'' to denote Greek indices contractions, while Latin indices contraction is naturally understood from matrix multiplication. }
A metric $g$ of signature $(1, 3)$ on $\M$ is induced from $\eta$ via  $\varpi$ according to $g(X, Y):=\eta\left( \theta(X), \theta(Y)\right)=\theta(X)^T\eta \theta(Y)$, or in a way more familiar to physicists $g:=e^T\eta e \rarrow g_{\mu\nu}={e_\mu}^a\eta_{ab} {e^b}_\nu$.

It should be noted that the gauge structure $(\P, \varpi)$ on $\M$ is not equivalent to a conformal class of metrics $c$. Indeed, the action of the local gauge group $\H_\text{\tiny{loc}}$ on $\varpi$ induces a conformal class of metrics via its soldering part, but the degrees of freedom of $\varpi$ compensated for by the gauge symmetry $\H_\text{\tiny{loc}}$ still amounts to more than $9=[c]$. 

But there is a way to make this Cartan geometry equivalent to a conformal manifold $(\M, c)$. In a way similar to the singling out of the Levi-Civita connection among all linear connections as the unique torsion-free and metric compatible connection, one can single out the so-called \emph{normal} conformal Cartan connection $\varpi_\text{\tiny{N}}$ as the unique one satisfying the constraints 
$\Theta=0$ (torsion free) and  ${W^a}_{bad}=0$.
 Together with the $\LieG_{-1}$-sector of the Bianchi identity $d\Omega+[\varpi, \Omega]=0$, these constraints imply $f=0$ (trace free), so that the curvature of the normal Cartan connection reduces to 
 $\Omega_\text{\tiny{N}}=\begin{psmallmatrix} 0 & C & 0  \\ 0 & W & C^t \\[0.5mm] 0 & 0 & 0 \end{psmallmatrix}$.
 From the normality condition ${W^a}_{bad}=0$  follows that $P$ 
 has components (in the $\theta$ basis of $\Omega^\bullet (\U)$)
 $P_{ab}=-\frac{1}{2} \left( R_{ab} - \frac{R}{6}\eta_{ab} \right)$,
 where $R$ and $R_{ab}$ are the Ricci scalar and Ricci tensor associated with the $2$-form $R=dA+A^2$. In turn, from this  follows that $W= R + \theta P + P^t\theta^t$ is the well known Weyl $2$-form. By the way, in the gauge $a=0$, $C:=dP +P A=DP$ looks like the familiar Cotton $2$-form.

 The gauge structure $(\P, \varpi_\text{\tiny{N}})$ is indeed equivalent to a conformal class of metric $c$ on $\M$. However, it would be hasty to  identify $A$ in $\varpi$ or $\varpi_\text{\tiny{N}}$ with the spin connection one is familiar with in physics, and by a way of consequence to take $R:=dA+A^2$ and $P$ as the Riemann and Schouten tensors. Indeed,  contrary to expectations $A$ is invariant under Weyl rescaling and neither $R$ nor $P$ have the known Weyl transformations. It turns out that one recovers the spin connection and the mentioned associated tensors only after  a dressing operation. See \cite{Attard-Francois2016}. 
\medskip

\paragraph{Spin representation} 

In the same way that there is a spin representation  $S\!O(1, 3) \xrightarrow{1:2} S\!L(2, \CC)$ of the Lorentz group, there is a spin representation of the conformal group $S\!O(2, 4) \xrightarrow{1:2}S\!U(2, 2)$, with the special unitary group $S\!U(2, 2)=\left\{ M \in GL_4(\CC)\ | \ M^*\b\Sigma M =\b\Sigma, \det M=1 \right\}$ with $\b\Sigma=\begin{psmallmatrix} 0 & \1 \\ \1 & 0 \end{psmallmatrix}$. The latter gives by restriction a representation of the structure group $H \xrightarrow{1:2} \b H$. We describe the latter in matrix notation as
\begin{align}
\label{IsoH-bH}
H=K_0K_1 \ \rarrow \ \b H=\b K_0 \b K_1 := \left\{ \begin{pmatrix}[1.2] z^{\sfrac{1}{2}} {\b S}^{-1*} & 0 \\ 0 & z^{-\sfrac{1}{2}} \b S\  \end{pmatrix}\begin{pmatrix}[1.2] \1 & -i \b r \\ 0 & \1 \end{pmatrix} \bigg|\ z\in\RR^*_+,\ \b S\in S\!L(2, \CC), 
\ \b r\in \text{Herm}(2, \CC)   \right\},
\end{align}
where we used \eqref{GroupMorph-Lorentz-SL} and the explicit isomorphism
\begin{align}
\label{LieMorph-Lorentz-SL2}
\RR^{4*} &\rarrow \text{Herm}(2, \CC),  \notag\\
r=x^t=x^T \eta  &\mapsto \b r:=x^0\s_0- x^i \s_i = \tfrac{1}{2}\begin{pmatrix} x^0-x^3 & -x^1 + ix^2 \\ -x^1 - ix^2 & x^0 + x^3 \end{pmatrix}. 
\end{align}
Using \eqref{MorphMinkowski}, \eqref{LieMorph-Lorentz-SL} and \eqref{LieMorph-Lorentz-SL2}, the Lie algebra morphism $\so(2, 4)=\LieG \rarrow \su(2, 2)=\b \LieG$ is then explicitely given by 
\begin{align}
\label{LieAlg-Iso}
\b\LieG =\b\LieG_{-1}+\b\LieG_0+\b\LieG_1= \left\{ 
\begin{pmatrix} -(\b s^* -\sfrac{\epsilon}{2}\1) &  -i\b\rho \\  i\b\tau  & \b s-\sfrac{\epsilon}{2}\1  \end{pmatrix}   
\bigg|\ \epsilon\in\RR, \b s\in \sl(2,\CC) \text{ and } \b \tau, \b\rho \in \text{Herm}(2, \CC)  \right\}
\supset
\b\LieH = \b\LieG_0+\b\LieG_1.
\end{align} 
The graded structure of $\b \LieG$ implies that 
\begin{align*}
[\b\LieG_{-1}, \b\LieG_1] \ni \left[ \begin{pmatrix} 0 & 0 \\ i\b\tau & 0  \end{pmatrix}, \begin{pmatrix} 0 & -i\b\rho \\ 0 & 0  \end{pmatrix} \right]=\begin{pmatrix} -\b\rho\b\tau & 0 \\ 0 & \b\tau \b\rho\end{pmatrix}=\begin{pmatrix} -\left(  (\b\tau\b\rho)_0^* +  \sfrac{\rho\tau}{2}\1\right) & 0 \\ 0 &  (\b\tau\b\rho)_0 + \sfrac{\rho\tau}{2}\1 \end{pmatrix} \in \b\LieG_0,
\end{align*}
so that $(\b\tau\b\rho)_0\in \sl(2, \CC)$ and $\rho\tau \in \RR$ is a scalar product. Clearly, $\b\LieG_{-1}$ and $\b\LieG_1$ are abelian. Relations \eqref{IsoH-bH} and \eqref{LieAlg-Iso}  reproduce in a handy and readable way some of the results of \cite{Klotz74}. 

The complex representation space for $G$ and $H$ is $\CC^4$. So, one may form the vector bundle $\sE=\P \times_{\b H}\CC^4$  associated to the Cartan bundle $\P(\M, H)$. A section of $\sE$ is a $\b H$-equivariant map on $\P$ whose local expression is 
\begin{align*}
\psi: \U \subset \M \rarrow \CC^4, \quad \text{given explicitely as column vectors }\quad \psi=\begin{pmatrix} \pi\\ \omega  \end{pmatrix}, \quad \text{ with }  \pi, \omega \in \CC^2 \text { dual Weyl spinors}.
\end{align*}
The covariant derivative induced by the Cartan connection is $\b D\psi=d\psi+\b\varpi\psi$, so that $\b D^2 \psi=\b \Omega \psi$. The spinorial conformal Cartan connection  and its curvature read
\begin{align*}
\b\varpi=\begin{pmatrix}[1.2]  -( \b A^* - \sfrac{a}{2}\1 )  &  -i\b P \\ i\b \theta & \b A -\sfrac{a}{2}\1  \end{pmatrix} \qquad \text{ and }  \qquad \b\Omega=\begin{pmatrix}[1.2]  -( \b W^* - \sfrac{f}{2}\1 )  &  -i\b C \\ i\b \Theta & \b W -\sfrac{f}{2}\1  \end{pmatrix}.
\end{align*}

The group metric $\b\Sigma$ naturally defines an invariant bilinear form on sections of $\sE$: given $\psi, \psi' \in \Gamma(\sE)$ one has
\begin{align*}
\langle \psi, \psi'  \rangle= \psi^* \b\Sigma \psi'=(\pi^*, \omega^*) \begin{pmatrix}   0 & \1 \\ \1 & 0 \end{pmatrix} \begin{pmatrix}\pi' \\ \omega' \end{pmatrix}= \pi^*\omega'+ \omega^*\pi'.
\end{align*} 
The covariant derivative $\b D$ naturally preserves this bilinear form since $\b \varpi$ is $\b \LieG$-valued: $\b D\b\Sigma=d\b\Sigma + \b \varpi^T \b\Sigma + \b\Sigma \b\varpi=0$.

\paragraph{Gauge transformations}

It would be tempting to identify $\sE$ with the twistor bundle. However its sections and covariant derivative thereof do not undergo  the defining Weyl transformation of a twistor as defined in section \ref{Bottom-up construction via prolongation of the Twistor Equation: a reminder}.
Indeed an element $\b \gamma$ of the local gauge group $\b\H=\b\K_0\b\K_1$ (we now drop the subscript ``loc'') can be factorized as $\b\gamma=\b\gamma_0\b\gamma_1: \U \rarrow \b H=\b K_0\, \b K_1$ with $\b\gamma_0\in \b\K_0 := \left\{ \gamma :\U\rarrow \b K_0  \right\}$ and $\b\gamma_1 \in \b\K_1:= \left\{  \b\gamma :\U\rarrow \b K_1 \right\}$.
Accordingly, through simple matrix calculations, the gauge transformations of $\psi$ w.r.t $\b\K_0$ and $\b\K_1$ are found to be
\begin{align}
\label{GT_psi}
\psi^{\b\gamma_0}= {\b\gamma_0}\-\psi \quad \rarrow \quad \begin{pmatrix} \pi^{\b\gamma_0} \\[1mm] \omega^{\b\gamma_0}  \end{pmatrix}=\begin{pmatrix} z^{\sfrac{1}{2}}\b S^* \pi \\[1mm]  z^{-\sfrac{1}{2}}\b S\- \omega \end{pmatrix}, \qquad \text{ and } \qquad 
\psi^{\b\gamma_1}= {\b\gamma_1}\-\psi \quad \rarrow \quad \begin{pmatrix} \pi^{\b\gamma_1} \\[1mm] \omega^{\b\gamma_1}  \end{pmatrix}=\begin{pmatrix} \pi + i\b r \omega \\[1mm] \omega \end{pmatrix}. 
\end{align}
 The same goes for $\b D\psi^{\b\gamma_0}$ and $\b D\psi^{\b \gamma_1}$. In the first relation put $\b S=\1$, compare with \eqref{GTtwistor} and notice the difference. It is clear that as it stands, $\sE$ is not the twistor bundle $\sT$ as previously defined. 
 
As for the Cartan connection, its  gauge transformation  w.r.t $\b\K_0$ is
\begin{align}
\label{GT_0}
\b\varpi^{\b\gamma_0}& =\b\gamma_0\-\b\varpi\b\gamma_0 + \b\gamma_0\-d\b\gamma_0 , \\[2mm]
\begin{pmatrix}[1.2]  -( \b A^* - \sfrac{a}{2}\1 )  &  -i\b P \\ i\b \theta & \b A -\sfrac{a}{2}\1  \end{pmatrix}^{\b\gamma_0}
&= \begin{pmatrix}[1.2]  -\left[ \left(\b S^*\b A^* \b S^{-1*}+d\b S \b S^{-1*} \right) -\sfrac{(a+z\-dz)}{2}\1 \right] & -i\ z\- \b S^*\b P \b S  \\  i\ z\b S\- \b \theta \b S^{-1*} & \b S\- \b A \b S + \b S\- d\b S - \sfrac{(a+z\-dz)}{2}\1 \end{pmatrix},  \notag
\end{align}
 and w.r.t $\b\K_1$ it reads
\begin{align}
\label{GT_1}
\b\varpi^{\b\gamma_1}& =\b\gamma_1\-\b\varpi\b\gamma_1 + \b\gamma_1\-d\b\gamma_1 , \\[2mm]
\begin{pmatrix}[1.2]  -( \b A^* - \sfrac{a}{2}\1 )  &  -i\b P \\ i\b \theta & \b A -\sfrac{a}{2}\1  \end{pmatrix}^{\b\gamma_1}
&= \begin{pmatrix}[1.2]   -\left[ \left( \b A^* + (\b r \b \theta)_0 \right) -\sfrac{(a-r\theta)}{2}\1 \right]  & -i\left[ \b P +d\b r - (\b r\b A + \b A^*\b r) + a\b r- \b r\b\theta \b r\right] \\ i\b\theta & \b A + (\b\theta\b r)_0 -\sfrac{(a-r\theta)}{2}\1 \end{pmatrix}. \notag
\end{align}
It is clear from the transformation of the soldering part $\b\theta$ of $\b\varpi$, that the metric induced by $\varpi^{\gamma_0}$ is $z^2g$. Thus the action of $\H$ on $\varpi$ induces a conformal class of metric $c$ on $\M$. But notice again that the Weyl transformations of $\b A$, $\b P$ are not as expected if we were to think of them as the  spin Levi-Civita connection and the Schouten tensor.

These disappointments will be corrected after the dressing field approach is performed on the \emph{normal} version of the gauge structures just described.

\subsubsection{Twistors from gauge symmetry reduction via dressing}  
\label{Twistors from gauge symmetry reduction via dressing}  

As we did for the tractor case, we aim at erasing the conformal boost gauge symmetry $\K_1$ through a dressing field. In \cite{Attard-Francois2016} we found such a dressing field
\begin{align*}
 u_1 :\U\rarrow K_1, \qquad \text{ that is } \quad u_1=\begin{pmatrix}  1 & q & \tfrac{1}{2}qq^t \\ 0 & \1 & q^t \\ 0 & 0 & 1 \end{pmatrix},
\end{align*}
 as solution of the gauge-like constraint $\Sigma(\varpi^{u_1}):=\Tr(A^{u_1} - a^{u_1})=-4 a^{u_1}=-4(a-q\theta)=0$ which, once solved for $q$, gives $q_a=a_\mu {e^\mu}_a$, or in index free notation $q=a\cdot e\-$.\footnote{Beware of the fact that in this index free notation $a$ is the set of components of the $1$-form $a$. This should be clear from the context.} We checked that $u_1$ satisfies the defining property of a $\K_1$-dressing field: $u_1^{\gamma_1}={\gamma_1}\- u_1$. 
 
Had we not known this, we could have found a $\b\K_1$-dressing field in the twistor context
\begin{align*}
\b u_1 : \U \rarrow \b K_1 \qquad \text{ that is  } \quad \b u_1=\begin{pmatrix} \1 & -i \b q \\ 0 & \1 \end{pmatrix},
\end{align*}
as solution of the gauge-like constraint $\Sigma'(\b\varpi^{\b u_1}):=\Tr(\b A^{u_1} - \sfrac{a^{u_1}}{2})=-(a-q\theta)=0$. This gives indeed the same $q \in \RR^{4*}$ as above which is then mapped to $\b q\in$ Herm$(2, \CC)$. Then, using \eqref{GT_1} we can check that $\b q^{\b\gamma_1}=\b q - \b r$, so that
 $\b u_1$ satisfies the defining property of a $\b\K_1$-dressing field: $\b u_1^{\b\gamma_1}={\b\gamma_1}\-\b u_1$. 
 
 Or, not ignoring the work done to find $u_1$ in the tractor case, we map it to $\b u_1$ thanks to the group morphism \eqref{IsoH-bH} which secures the fact that the defining property is respected. And we are done. 
 \medskip
 
 With this $\b\K_1$-dressing field we can apply - the local version of - proposition \ref{P1} and form the $\b\K_1$-invariant composite fields
\begin{align}
\label{CompFields_1}
\b\varpi_1:\!&=\b\varpi^{\b u_1}\!=\b u_1 \- \b\varpi \b u_1 + \b u_1\-d\b u_1\!=\begin{pmatrix}[1.2]  - \b A_1^*   &  -i\b P_1 \\ i\b \theta & \b A_1  \end{pmatrix}
, \quad \b\Omega_1:=\b\Omega^{\b u_1}\!=\b u_1\-\b\Omega \b u_1=d\b \varpi_1+\b \varpi_1^2=
\begin{pmatrix}[1.2]  -( \b W_1^* - \sfrac{f_1}{2}\1 )  &  -i\b C_1 \\ i\b \Theta & \b W_1 -\sfrac{f_1}{2}\1  \end{pmatrix} \notag \\[1mm]
\psi_1:&=\b u_1\-\psi=\begin{pmatrix} \pi_1\\ \omega_1  \end{pmatrix}, 
\qquad \text{and}\qquad
 \b D_1\psi_1=d\psi_1+\b\varpi_1 \psi_1=\begin{pmatrix}[1.2] d\pi_1  -\b A_1^*\pi_1 -i\b P_1 \omega_1 \\ d\omega_1 + \b A_1\omega_1 + i\b\theta \pi_1 \end{pmatrix}=\begin{pmatrix}[1.2] \nabla_1\pi_1 -i\b P_1 \omega_1 \\ \nabla_1\omega_1 + i\b\theta \pi_1  \end{pmatrix}
\end{align}
As is usual ${\b D_1}^2 \psi_1 = \b\Omega_1 \psi_1$.  We notice that $f_1 = -2\Tr\left(\b\theta\b P_1 \right)=P_1\w \theta$ is the antisymmetric part of the tensor $P_1$.

The claim is twofold. First, we assert that $\psi_1$ is a twistor and that the covariant derivative $b D_1$ induced from the dressed spin conformal Cartan connection $\b\varpi_1$ is a ``generalized'' twistor connection. Second, the composite fields \eqref{CompFields_1} are gauge fields of a non-standard kind - such as described in section \ref{The composite fields as  a new kind of gauge fields} - w.r.t Weyl symmetry, but genuine gauge fields - according to section \ref {The composite fields as genuine gauge fields} - w.r.t Lorentz symmetry. Both assertions are supported by the analysis of the residual gauge transformations of these composite fields.

\paragraph{Residual gauge symmetries} Being $\b\K_1$-invariant  by construction, the composite fields \eqref{CompFields_1} are expected to display a $\b\K_0$-residual gauge symmetry. This group breaks down as a direct product of the (spin) Lorentz and Weyl groups: $\b K_0=\mathsf{SL}(2, \CC) \times \mathsf W$, where $\mathsf{SL}(2, \CC)\simeq S\!L(2, \CC) \oplus S\!L(2, \CC) ^*$. We focus on Lorentz symmetry first, then only bring our attention to Weyl symmetry. Here again we could use the results found in \cite{Attard-Francois2016} and map them via the isomorphism \eqref{IsoH-bH}, but to be complete we indicate how to reach the same results from first principles. 
\smallskip

The residual gauge transformations of the composite fields \eqref{CompFields_1} under the spin Lorentz gauge group, defined as 
$\SL:=\left\{ \mathsf S:\U \rarrow \mathsf{SL}(2, \CC)\ |\ \mathsf{S}=\b\gamma_{0|z=1},  {\mathsf S}^{\mathsf S'}={\mathsf S'}\- \mathsf S \mathsf S'  \right\}$, are inherited from that of the dressing field $\b u_1$. Using \eqref{GT_0} to compute $\b q^{ \mathsf S}=a^{ \mathsf S}\cdot (\b e^{ \mathsf S})\-=\b S^*\b q\b S$, one easily finds that $\b u_1^\mathsf S=\mathsf S\- \b u_1 \mathsf S$. This is a local instance of Proposition \ref{Prop2}, which then allows to conclude that the composite fields 
 are \emph{genuine} gauge fields w.r.t Lorentz gauge symmetry. Hence, from Corollary \ref{Cor3} follows that their residual $\SL$-gauge transformations are
\begin{align}
\label{CompFields_1_S}
&\b\varpi_1^\mathsf S={\mathsf S}\-  \b\varpi_1 \mathsf S  + { \mathsf S}\- d \mathsf S=
\begin{pmatrix}[1.2] -\left(  \b S^*\b A_1 {\b S}^{-1*} + d\b S^*{\b S}^{-1*} \right) &  -i\ \b S^* \b P_1 \b S \\  i\ {\b S}\-\b\theta {\b S}^{-1*} &  {\b S}\-\b A_1 \b S + {\b S}\- d\b S &  \end{pmatrix},
  \notag \\[1mm]
 &\b\Omega_1^\mathsf S={ \mathsf S}\- \b\Omega_1  \mathsf S=
 \begin{pmatrix}[1.2]  -\b S^*( \b W_1^* - \sfrac{f_1}{2}\1 ){\b S}^{-1*}  &  -i \b S^* \b C_1 \b S \\ i\  {\b S}\-\b\Theta {\b S}^{-1*} & {\b S}\-\left( \b W_1 -\sfrac{f_1}{2}\1 \right) \b S  \end{pmatrix},   \notag \\[1mm]
&\psi_1^\mathsf S= {\mathsf S}\- \psi_1= \begin{pmatrix} \b S^* \pi_1 \\ {\b S}\- \omega_1  \end{pmatrix},\qquad \text{ and } \qquad (\b D_1\psi_1)^\mathsf S=\b D_ 1^\mathsf S \psi_1^\mathsf S = {\mathsf S}\- \b D_1\psi_1.
\end{align}
We notice in particular that $\psi_1$ behaves as a standard section of a $\mathsf{SL}(2, \CC)$-associated bundle: $\sE_1\!=\sE^{\b u_1}\!=\P\times_\mathsf{SL} \CC^4$. 

We repeat the analysis for the Weyl symmetry. The residual transformations of the composite fields under the Weyl gauge group, that we define as $\W:=\left\{ Z:\U \rarrow \mathsf W\ |\  Z=\b\gamma_{0|S=\1}, Z^{Z'}=Z  \right\}$, are inherited from that of the dressing field $\b u_1$. Using \eqref{GT_0} to compute $\b q^Z=a^Z\cdot (\b e^Z)\-= z\- ( \b q + \b \Upsilon)$, with $\Upsilon:=z\-\d z\cdot e\- \rarrow \Upsilon_a=z\-\d_\mu z \ {e^\mu}_a$, one easily finds that 
\begin{align}
\label{WeylGT_u_1}
&\b u_1^Z=Z\- \b u_1 C(z)\quad  \text{ where the map }\quad  C: W \rarrow \b K_1 \mathsf W \subset \b H \quad \text{ is defined by }\\[1mm]
&C(z):=\b k_1(z)Z =\begin{pmatrix} \1 & -i\b \Upsilon \\ 0 & \1\end{pmatrix} \begin{pmatrix}z^{\sfrac{1}{2}}\1 & 0 \\ 0 &  z^{-\sfrac{1}{2}}\1 \end{pmatrix}= \begin{pmatrix}[1.2] z^{\sfrac{1}{2}} \1 &  -i\ z^{-\sfrac{1}{2}}\b \Upsilon \\ 0 & z^{-\sfrac{1}{2}}\1 \end{pmatrix}. \notag
\end{align}
 Notice that contrary to  genuine gauge group members, elements of type $C(z)$  do not form a group: $C(z)C(z')\neq C(zz')$. 
 Actually \eqref{WeylGT_u_1} is a local instance of Proposition \ref{P4} with $C$ a  $1$-$\alpha$-cocycle satisfying Proposition \ref{alpha_cocycle}. Indeed one can check that  $C(zz')=C(z'z)=C(z')\ {Z'}\- C(z) Z'$, which is the defining property of an abelian $1$-$\alpha$-cocycle. Furthermore, under a further $\W$-gauge transformation and due to $\b e^Z=z\b e$, one has $\b k_1(z)^{Z'}= {Z'}\- \b k_1(z)Z'$, which implies $C(z)^{Z'}= {Z'}\- C(z) Z'$. So if $\b u_1$ undergoes a a further $\W$-gauge transformation we have
 \begin{align*}
 \left(\b u_1^Z\right)^{Z'}=\left(Z^{Z'}\right)\- \b u_1^{Z'} C(z)^{Z'}=Z\-\ {Z'}\- \b u_1C(z')\ {Z'}\- C(z) Z'= (ZZ')\- \b u_1 C(zz'). 
 \end{align*}
 
 This implies that the composite fields \eqref{CompFields_1} are indeed instances of gauge fields of the new kind described in section \ref{The composite fields as  a new kind of gauge fields}. As a consequence, by Proposition \ref{P5} we have that their residual $\W$-gauge transformations  are
\begin{align}
\b\varpi_1^Z&=C(z)\- \b\varpi_1 C(z) + C(z)\-dC(z)= 
\begin{pmatrix}[1.2] -\b A^*_1 -(\b\Upsilon\b\theta)_0\ \  &  -i\ z\-\left[\b P_1  + \left(d\b\Upsilon -\b\Upsilon \b A_1 - \b A_1^* \b\Upsilon\right) -\b\Upsilon\b\theta\b\Upsilon \right]    \\ i\ z\b\theta  &   \b A_ 1 + (\b\theta\b\Upsilon)_0  \end{pmatrix},
  \label{varpi_1_Z}    \\[1mm]                     
\b\Omega_1^Z&=C(z)\- \b\Omega_1 C(z)=
\begin{pmatrix}[1.2] -\left( \b W_1^* -\sfrac{f_1}{2}\1 \right) - \b\Upsilon\b\Theta\ \ &   \ -i\ z\-\left( \b C_1 - \b\Upsilon\b W_1 - \b W_1^*\b\Upsilon + f_1\1 \b\Upsilon - \b\Upsilon\b\Theta\b\Upsilon  \right)   \\ i\ z\b\Theta &   \b W_1 -\sfrac{f_1}{2} \1 + \b\Theta\b\Upsilon   \end{pmatrix},         
  \label{Omega_1_Z}  \\[1mm]                    
\psi_1^Z&=C(z)\- \psi_1 = \begin{pmatrix}[1.2] z^{-\sfrac{1}{2}}\left( \pi_1 + i\b\Upsilon \omega_1\right) \\ z^{\sfrac{1}{2}} \omega_1 \end{pmatrix}, \qquad \text {and }\quad (\b D_1\psi_1)^Z=\b D_1^Z\psi_1^Z= C(z)\- \b D_1\psi_1.        \label{Twistor_Connection_1}
\end{align}

 First, notice that in \eqref{varpi_1_Z} now the Lorentz part $\b A_1$ of the composite field $\b\varpi_1$ indeed exhibits the known Weyl transformation for the spin connection and that $\b P_1$ transforms as the  Schouten tensor (in an orthonormal basis). But actually the former genuinely reduce to the latter only when one restricts to the  \emph{normal} case $\b\varpi_{\n, 1}$, so that $\b A_1$ is a function of $\b\theta$ and $\b P_1=\b \sP_1$ is symmetric and a function of $\b A_1$. So $f_1$ vanishes and we have
\begin{align}
\label{normal}
\b\Omega_{\text{\tiny N},1}=d\varpi_{\text{\tiny N},1}+\varpi_{\text{\tiny N},1}^2=\begin{pmatrix}[1.2]  -\b \sW_1^*   &  -i\b \sC_1 \\ 0 & \b \sW_1 \end{pmatrix}, 
\ \ \text{  and  } \ \ \b\Omega_{\text{\tiny N},1}^Z=C(z)\- \b\Omega_{\text{\tiny N},1}C(z)=\begin{pmatrix}[1.2]  -\b \sW_1^*   &  -i \ z\-\!\left(\b \sC_1- \b\Upsilon\b W_1 - \b \sW_1^*\b\Upsilon \right) \\ 0 & \b \sW_1 \end{pmatrix}.
\end{align}
We  see that $\b \sC_1=\nabla_1 \b\sP_1:=d\b\sP_1 +\b\sP_1 \b A_1 -\b A_1^* \b\sP_1$ is the Cotton tensor - and indeed transforms as such - while $\b \sW_1$ is the invariant Weyl tensor. 
\medskip

Notice that the first relation in \eqref{Twistor_Connection_1} is - modulo the $z$ factors - \eqref{GTtwistor}. So the dressed section $\psi_1$ is identified with a twistor field, section of a $C$-vector bundle  $\sE_1\!=\sE^{\b u_1}\!=\P \times_{C( W)}\CC^4$. 
The invariant bilinear form on $\sE$ defined by the group metric $\b\Sigma$  is also defined on $\sE_1$: $\left\langle \psi_1, \psi'_1  \right\rangle= \psi_1^T \b\Sigma \psi'_1$. Indeed since $C(z) \in \b K_1\mathsf W \subset \b H$, we have $\left\langle \psi_1^Z, {\psi'_1}^Z \right\rangle = \left\langle C(z)\-\psi_1, C(z)\- \psi'_1 \right\rangle = \psi_1^* (C(z)\-)^* \b\Sigma C(z)\- \psi_1'=\psi_1^*\b\Sigma \psi'_1 = \left\langle \psi_1, \psi_1' \right\rangle$.

Moreover, $\b D_1:=d + \b\varpi_1$ in \eqref{CompFields_1} is a generalization of the twistor connection \eqref{TE4}. But then the term ``connection'', while not inaccurate, could hide the fact that $\b\varpi_1$ is no more a standard connection w.r.t Weyl symmetry. So we shall prefer to call $\b D_1$ a  generalized twistor \emph{covariant derivative}. 
The usual twistor covariant derivative \eqref{TE4} is recovered by restriction to the dressing of the normal Cartan connection: $\b D_{\text{\tiny N},1}\psi_1=d\psi_1 + \b\varpi_{\text{\tiny N},1} \psi_1$. Then we have that  $\b D_{\text{\tiny N},1}^2\psi_1=\b\Omega_{\text{\tiny N},1} \vphi_1$ is just \eqref{twistor_curvature}.  We note that $\b\varpi_1$ being $\b\LieG$-valued, $\b D_1\b\Sigma=0$  and $\b D_1$ preserves the bilinear form $\langle\ \!, \rangle$. 

In short, by erasing via dressing the $\b\K_1$-gauge symmetry from the conformal Cartan gauge structure $\left( (\P, \varpi), \sE\right)$ over $\M$, we have recovered top-down the twistor bundle and twistor covariant derivative, $(\sE_1, \b D_{\text{\tiny N},1})$, as a special case of the $C$-vector bundle endowed with a covariant derivative $(\sE_1, \b D_1)$.  The link between the normal conformal Cartan connection and the twistor covariant derivative as already been noticed by Friedrich \cite{Friedrich77}, whose result we thus generalize.

The actions of $\SL$ and $\W$ on the composite fields $\chi_1$ are compatible and commute. Indeed, we have first that $\sS^\W=\sS$ so that on the one hand: $\left(\chi_1^\SL\right)^\W=\left(\chi_1^\mathsf S\right)^\W=\left( \chi_1^\W \right)^{\mathsf S^\W}=\left( \chi_1^{C(z)}\right)^\mathsf S=\chi_1^{C(z)\mathsf S}$. But then we also have $C(z)^\SL={\mathsf S}\- C(z)\mathsf S$, so on the other hand we get: $\left( \chi_1^\W \right)^\SL=\left( \chi_1^{C(z)}\right)^\SL=\left( \chi_1^\SL\right)^{C(z)^\SL}=\left(\chi_1^\mathsf S \right)^{{\mathsf S}\- C(z)\mathsf S}=\chi_1^{C(z)\mathsf S}$.
We should then refine our notation for the bundle $\sE_1$ and write $\sE_1=\P\times_{ C(\mathsf W)\ \! \mathsf{SL}} \CC^4$.
\smallskip

As is suggested by the considerations at the end of section \ref{The composite fields as genuine gauge fields}, the fact that the composite fields \eqref{CompFields_1} are genuine $\SL$-gauge fields satisfying \eqref{CompFields_1_S} entitles us to ask if a further dressing operation aiming at erasing Lorentz symmetry is possible. In \cite{Attard-Francois2016} we showed that in the case of tractors, the vielbein could be used to this purpose. But since there is no finite dimensional spin representation of $GL$, one  suspects that in the case at hand the vielbein cannot be used. This is indeed so. Furthermore, one just has to look at the $S\!L(2, \CC)$ gauge transformation of the vielbein to see that it is unsuited as a dressing field. So our process of symmetry reduction ends here.

\subsection{BRST treatment}
\label{BRST treatment}

The spin representation of the gauge group of the initial Cartan geometry is $\b\H$, so the associated ghost $\b v\in$ Lie$\b\H$ splits along the grading of $\b\LieH$,
\begin{align}
\b v=\b v_0+\b v_\rho=\b v_\epsilon+\b v_\ss+\b v_\rho=\begin{pmatrix} \sfrac{\epsilon}{2} & 0 \\  0 & \sfrac{\epsilon}{2} \end{pmatrix} + \begin{pmatrix} -\b s^* & 0 \\ 0  & \b s \end{pmatrix}+\begin{pmatrix} 0 & -i \b \rho \\ 0 & 0 \end{pmatrix}.
\end{align} 
The BRST operator splits accordingly as $s=s_0+s_1= s_\ww+s_\l+s_1$. Then the BRST algebra for the gauge fields $\chi=\{ \b\varpi, \b\Omega, \psi \}$ is
\begin{align}
s\b\varpi=-\b Dv=-d\b v -[\b\varpi, \b v], \qquad s\b\Omega=[\b\Omega, \b v], \qquad s \psi=-\b v\psi \qquad \text{and} \qquad s\b v=-{\b v}^2,
\end{align}
with the first and third relations in particular reproducing the infinitesimal versions of \eqref{GT_0}-\eqref{GT_1} and \eqref{GT_psi}. Denote by $\mathsf{BRST}$ this initial algebra.
As the general discussion of section \ref{Application to the BRST framework} showed, the dressing approach modifies it. 
\medskip

 We know from this  general discussion that the composite fields $\chi_1=\{ \b\varpi_1,\b\Omega_1, \psi_1 \}$ satisfy a modified BRST algebra formally similar but with composite ghost $\b v_1:=\b u_1\- \b v \b u_1 + \b u_1\-s\b u_1$. The inhomogeneous term can be found explicitly from the finite gauge transformations of $\b u_1$. Writing the linearizations (where the linear parameters are turned into ghosts) $\b\gamma_1\simeq\1 + \b v_\rho$ and $\sS \simeq\1 + \b v_\ss$, the BRST actions of $\b\K_1$ and $\SL$ are easily found to be  
\begin{align*}
\b u_1^{\b\gamma_1}=\b\gamma_1\-\b u_1 \quad \rarrow \quad s_1\b u_1=-\b v_\rho \b u_1  \qquad \text{and} \qquad \b u_1^\sS=\sS\- \b u_1 \sS \quad \rarrow \quad s_\l \b u_1=[\b u_1, \b v_\ss].
\end{align*}
This shows that the Lorentz sector gives an instance of the general result \eqref{NewGhost2}. Now, let us define the linearizations  $Z\simeq\1 + \b v_\epsilon$ and $\b k_1(z)\simeq\1 + \b k_1(\epsilon)$, so that $C(z)=\b k_1(z)Z\simeq\1+c(\epsilon)=\1 + \b k_1(\epsilon)+\b v_\epsilon$ where $\b k_1(\epsilon):=\begin{psmallmatrix} 0 & -i\b{\d\epsilon} \\ 0 & 0\end{psmallmatrix}$ with $\d\epsilon:=\d_\mu \epsilon\ {e^\mu}_a$. 
The BRST action of $\W$ is
\begin{align*}
\b u_1^Z=Z\-\b u_1 C(z) \quad \rarrow \quad s_\ww \b u_1=-\b v_\epsilon \b u_1 + \b u_1 c(\epsilon).
\end{align*}
This shows that the Weyl sector gives an instance of the general result \eqref{NewGhost3}. We then get the composite ghost
\begin{align}
\b v_1:&=\b u_1\- (\b v_\epsilon+\b v_\ss+\b v_\rho)  \b u_1 + \b u_1\- (s_\ww+s_\l +s_1) \b u_1,  \notag \\
	   &=\b u_1\- (\b v_\epsilon+\b v_\ss+\b v_\rho) \b u_1 + \b u_1\-  \big( -\b v_\epsilon \b u_1 + \b u_1 c(\epsilon) \ + \ [\b u_1, \b v_\ss] \ - \ \b v_\rho \b u_1 \big), \notag\\
	&=c(\epsilon) + \b v_\ss=\begin{pmatrix} -\left(\b s^* -\sfrac{\epsilon}{2}\1 \right) & -i\b{\d\epsilon} \\ 0 & \b  s - \sfrac{\epsilon}{2} \end{pmatrix}.
\end{align} 
We see that the ghost of conformal boosts $\b\rho$ has disappeared from this new ghost. This means that $s_1 \chi_1=0$, which reflects the $\b\K_1$-gauge invariance of the composite fields $\chi_1$. The composite ghost $\b v_1$ only depends on $\b v_\ss$ and $\epsilon$, it encodes the residual $\b \K_0$-gauge symmetry.
The BRST algebra for the composite fields $\chi_1$ is then explicitly
\begin{align*}
s\b\varpi_1&=-\b D_1\b v_1=-d\b v_1 -[\b\varpi_1, \b v_1]
=\begin{pmatrix}[1.2] \nabla^*\b s^* - (\b{\d\epsilon}\b\theta)_0 &  -i \left( -\epsilon \b P_1-\nabla(\b{\d\epsilon}) - (\b P_1 \b s - \b s^* \b P_1) \right) \\ i \left( -\epsilon\b\theta + \b s\b\theta-\b\theta \b s^* \right) & -\nabla \b s -(\b\theta \b{\d\epsilon})_0 \end{pmatrix},
\end{align*}
where  $\nabla \b s=d\b s + [\b A_1, \b s]$, $\nabla^* \b s^*=d\b s^* - [\b A_1^*, \b s^*]$ and $\nabla(\b{\d\epsilon})=d\b{\d\epsilon} + \b{\d\epsilon} \b A_1- \b A_1^* \b{\d\epsilon}$.
\begin{align*}
 &s\b\Omega_1=[\b\Omega_1, \b v_1]=\begin{pmatrix}[1.2] - \b{\d\epsilon}\b\Theta+  [\b W_1^*,\b s^*] & -i\left(  -\epsilon \b C_1 + (\b{\d\epsilon} \b W_ 1 - \b W_1^*\b{\d\epsilon} ) + (\b C_1 \b s - \b s^*\b C_1) \right)  \\ 
 i\left( -\epsilon\b\Theta - \b s\b\Theta - \b\Theta\b s^*  \right) & [\b W_1, \b s] + \b\Theta \b{\d\epsilon}\end{pmatrix}, \\[1mm]
&\text{in the normal case }\quad s\b\Omega_{\n, 1}=\begin{pmatrix}[1.2]   [\b W_1^*,\b s^*] & -i\left(  -\epsilon \b C_1 + (\b{\d\epsilon} \b W_ 1 - \b W_1^*\b{\d\epsilon} ) + (\b C_1 \b s - \b s^*\b C_1) \right)  \\ 
0 & [\b W_1, \b s] \end{pmatrix},\\[1mm]
&s\psi_1=-\b v_1\psi_1 =\begin{pmatrix}[1.2](\b s^* -\sfrac{\epsilon}{2}\1) \pi_1 + i\b{\d\epsilon}\omega_1 \\ -(\b s - \sfrac{\epsilon}{2}\1) \omega_1\end{pmatrix}, 
\qquad \text{and} \qquad s\b v_1=-{\b v_1}^2=\begin{pmatrix}  \b s^* \b s^* & -i\left(\b{\d\epsilon} \b s - \b s^*\b{\d\epsilon}\right) \\ 0 & \b s\b s  \end{pmatrix}.
\end{align*}
Denote this algebra $\mathsf{BRST}_{\ww, \l}$. Since $\b v_1=c(\epsilon)+\b v_\ss$, it splits naturally as a Lorentz and a Weyl subalgebras, $s=s_\ww+s_\l$. The Lorentz algebra $\mathsf{BRST}_\l$
, obtained by setting $\epsilon=0$, shows the composites fields $\chi_1$ to be genuine Lorentz gauge fields (compare with \eqref{CompFields_1_S}). The Weyl algebra $\mathsf{BRST}_\ww$
, obtained by setting $\b s=0$, shows $\chi_1$ to be non-standard Weyl gauge fields (compare with  \eqref{varpi_1_Z}-\eqref{Twistor_Connection_1}).

\subsection{The Yang-Mills Lagrangian of the spin conformal Cartan  connection and Weyl gavity}
\label{The Yang-Mills Lagrangian of the spin conformal Cartan connection}

In this section we first highlight the fact that the natural Yang-Mills Lagrangian for the twistor $1$-form $\b\varpi_{\n, 1}$ actually reproduces Weyl gravity. This gives a geometrically clear interpretation of the results of Merkulov \cite{Merkulov1984_I} and shows that local twistors and the twistor covariant derivative  plainly define a Yang-Mills theory (see the foonote p. $133$ of \cite{Penrose-Rindler-vol2}). We  also argue that the dressed spin conformal Cartan connection  $\b\varpi_1$ is almost the ``modified twistor connection''  introduced in \cite{Merkulov1984_II} where is advocated an approach to describe both Weyl gravity and Electromagnetism as a unified conformally invariant theory. But then we have to higlight a drawback of this approach.
\medskip

First let us define the invariant Killing forms for $\su(2, 2)$ and $\so(2, 4)$. Given $\b A, \b B \in \su(2, 2)$, the Killing form is $B_{\su(2,2)}(\b A, \b B):=\tfrac{1}{2}\left(  \Tr(\b A\b B) + \Tr(\b B^*\b A^*)\right)$. Given $A, B \in \so(2, 4)$, the Killing form is $B_{\so(2, 4)}(A, B):=\Tr(AB)$. The same formulae hold for $\sl(2, \CC)$ and $\so(1, 3)$ and define $B_{\sl(2, \CC)}$ and $B_{\so(1, 3)}$

\paragraph{Twistorial conformal gravity}
In a previous work \cite{AFL2016_I} we showed that the Yang-Mills Lagrangian associated to the dressed normal Cartan connection $\varpi_{\n, 1}$ reproduces Weyl gravity,
\begin{align*}
L_\text{YM}(\varpi_{\n, 1}):= \tfrac{1}{2}B_{\so(2, 4)}(\Omega_{\n,1},  * \Omega_{\n, 1})=  \tfrac{1}{2}B_{\so(1, 3)}(\sW_1, *\sW_1)= \tfrac{1}{2}\Tr(\sW_1 \w *\sW_1)=L_\text{Weyl}(\theta),
\end{align*}
where $*$ is the Hodge star operator on differential forms on $\M$, and $\sW_1$ is the Weyl tensor.
It is a fact that the only true degrees of freedom of $\varpi_{\n, 1}$ are those of the vierbein $\theta$ since in the normal case $A_1$ is expressed as a function of $\theta$, and $P_1$ as a function of $A_1$ through $R_1=dA_1 +{A_1}^2$. So it is no surprise that the field equations obtained on the one hand by varying the action $S_\text{YM}(\varpi_{\n, 1})$ w.r.t $\varpi_{\n, 1}$, and on the other hand by varying the action $S_\text{Weyl}(\theta)$ w.r.t $\theta$, should coincide. In the first case we obtain the Yang-Mills equation for the (dressed) normal conformal Cartan connection, and in the second case we obtain the Bach equation
\begin{align*}
\frac{\delta S_\text{YM}(\varpi_{\n, 1})}{\delta\varpi_{\n, 1}}=0 \quad \rarrow \quad D_1*\Omega_{\n, 1}=0    \quad\qquad \Leftrightarrow \quad\qquad \frac{\delta S_\text{Weyl}(\theta)}{\delta\theta}=0 \quad \rarrow\quad  B_{ab}=0,
\end{align*}
where $B_{ab}$ is the Bach tensor. The equivalence of the field equations was first noticed by Korzy{\' n}ski and Lewandowski \cite{Korz-Lewand-2003}. The origin of this equivalence is clear from our perspective.

Now, consider  the normal spin conformal Cartan connection $\b\varpi_\n$ as a $\b\H$-gauge potential, the associated Yang-Mills Lagrangian is 
\begin{align*}
L_\text{YM}(\b\varpi_\n):=\tfrac{1}{4}B_{\su(2, 2)}(\b\Omega_\n, * \b\Omega_\n)= \tfrac{1}{2} B_{\sl(2, \CC)}(\b W, *\b W)
\end{align*}
 Using proposition \ref{Prop-Lagrangian} of section \ref{Local aspects and Physics}, we know that since $\b u_1:\U \rarrow \b K_1$ we have that 
\begin{align*}
L_\text{YM}(\b\varpi_\n)=L_\text{YM}(\b\varpi_{\n, 1})=\tfrac{1}{4}B_{\su(2, 2)}(\b\Omega_{\n,1},  * \b\Omega_{\n, 1})=\tfrac{1}{2} B_{\sl(2, \CC)}(\b \sW_1, *\b \sW_1)
\end{align*}
In other words, the actual symmetry of the theory is not $\b\H=\b\K_0\b\K_1$ with gauge potential $\varpi_\n$, but $\b\K_0$, with gauge potential the twistor $1$-form $\b\varpi_{\n, 1}$. 
By the way, since $\sl(2,\CC) \simeq \so(1, 3)$,  their Killing forms must coincide. As a matter of fact, for $A, B \mapsto \b A, \b B$ one has  $B_{\sl(2, \CC)}(\b A, \b B)=B_{\so(1, 3)}(A, B)$. This means that $L_\text{YM}(\b\varpi_\n)=L_\text{Weyl}(\theta)$. So, as above, we have equivalence between the Yang-Mills equation for the twistor $1$-form and the Bach equation
\begin{align*}
\frac{\delta S_\text{YM}(\b\varpi_{\n, 1})}{\delta\b\varpi_{\n, 1}}=0 \quad \rarrow \quad \b D_1*\b\Omega_{\n, 1}=0    \quad\qquad \Leftrightarrow \quad\qquad \frac{\delta S_\text{Weyl}(\theta)}{\delta\theta}=0 \quad \rarrow\quad  B_{ab}=0,
\end{align*}

This equivalence was first noticed by Merkulov \cite{Merkulov1984_I} and deemed surprising. We see that it is quite natural from our perspective, according to which one just has to observe that the twistor $1$-form $\b\varpi_{\n, 1}$ is the spin version of the dressed normal conformal Cartan connection $\varpi_{\n, 1}$ whose only true degrees of freedom are those of the vierbein $\theta$.

\paragraph{Attempt at a twistorial unification of conformal gravity and electromagnetism}
In our notation the usual twistor $1$-form is $\b\varpi_{\n, 1}=\begin{psmallmatrix}  -\b A_1^* & -i \b \sP_1 \\ \theta & \b A_1 \end{psmallmatrix}$, with $\b A_1$ the spin Levi-Civita connection and $\b \sP_1$ the symmetric Schouten tensor. In \cite{Merkulov1984_II} a twistorial unification of Weyl gravity and electromagnetism  was proposed. To do so,  a modification of the twistor $1$-form $\b\varpi_{\n,1}$  was suggested, according to which to $\b\sP_1$ is added  a Weyl invariant antisymmetric component $\b f_1$, interpreted as the Maxwell-Faraday tensor. Let us write $\b\sP_1+\b f_1:=\b P_1$ so that the modified twistor $1$-form is 
\begin{align}
\label{Merkulov-version} 
\b\varpi'_1=\begin{pmatrix}  -\b A_1^* & -i \b P_1 \\ \theta & \b A_1 \end{pmatrix}, 
\quad \text{ with curvature }\quad
\b\Omega'_1=d\b\varpi_1 +{\b\varpi_1}^2=\begin{pmatrix} -\left( \b\sW_1^* -\sfrac{f_1}{2}\1 \right) &  -i\b C_1 \\ 0 & \b \sW_1 - \sfrac{f_1}{2} \end{pmatrix},
\end{align}
where $f_1=P_1 \w \theta$ is the antisymmetric part of $P_1$, and $\b C_1= \b\sC_1 + \nabla_1\b f_1$. From there, a natural Yang-Mills Lagrangian is proposed
\begin{align}
\label{Lag-Merk}
L_\text{YM}(\b\varpi'_1):=\tfrac{1}{4}B_{\su(2, 2)}(\b\Omega'_1,  * \b\Omega'_1)=\tfrac{1}{2} B_{\sl(2, \CC)}(\b \sW_1, *\b \sW_1) +\tfrac{1}{4}f_1\w *f_1.
\end{align}
It is indeed Lorentz and Weyl invariant since the curvature transforms as
\begin{align*}  
\b\Omega_1^{'\mathsf S}={ \mathsf S}\- \b\Omega'_1  \mathsf S=
 \begin{pmatrix}[1.2]  -\left(\b S^* \b \sW_1^*{\b S}^{-1*}  - \sfrac{f_1}{2}\1 \right)  &  -i \b S^* \b C_1 \b S \\ 0 & \left({\b S}\- \b W_1 \b S -\sfrac{f_1}{2}\1 \right)   \end{pmatrix},
\\
\b\Omega_1^{'Z}=C(z)\- \b\Omega'_1 C(z)=\begin{pmatrix}[1.2]  -\left(\b \sW_1^* -\sfrac{f_1}{2}\1 \right)  &  -i \ z\-\!\left(\b C_1- \b\Upsilon\b W_1 - \b \sW_1^*\b\Upsilon+f_1\1 \b\Upsilon \right) \\ 0 & \b \sW_1-\sfrac{f_1}{2}\1 \end{pmatrix}
\end{align*}
\medskip
which are special cases of \eqref{CompFields_1_S} and \eqref{Omega_1_Z}. Variation of $S_\text{YM}(\b\varpi'_1)$ w.r.t $\b\varpi'_1$ gives the Yang-Mills equation with source $\b D_1 *\b\Omega'_1=\kappa T$,  where $\kappa$ is an irrelevant constant and the source term is the (Hodge dual of) the energy-momentum tensor of the electromagnetic field $T_{ab}=\tfrac{1}{4}f_{cd}f^{cd}\eta_{ab}+f_{bc}f^{cd}\eta_{da}$. The field equation infolds as the source-free dynamical Maxwell equation, $\nabla^af_{ab}=0$, and the Weyl gravity  coupling to the electromagnetic field, $B_{ab}=\kappa T_{ab}$. 

All this would be satisfying is not for a hidden flaw. Notice that \eqref{Merkulov-version} is sort of in between the dressed spin conformal Cartan connection $\b\varpi_1$ \eqref{CompFields_1} and its normal version which is the twistor $1$-form $\b\varpi_{\n, 1}$ \eqref{normal}. But are we free to posit such an intermediate object? Actually no. Indeed, by requiring $\b\varpi'_1$ to be torsion free as in \eqref{Merkulov-version}, one an check that the $\b\LieG_{-1}$-sector of the  Bianchi identity  $\b D'_1\b\Omega'_1=0$ gives $\b\theta \sW_1^*+\sW_1\b\theta=f_1\b\theta$. The latter  implies $f_1=0$ by virtue of the fact that the Ricci contraction of the Weyl tensor vanishes, ${{\sW_1}^a}_{bac}=0$. Then the modified twistor $1$-form $\b\varpi'_1$ reduces to the standard one $\b\varpi_{\n, 1}$, the curvature $\b\Omega'_1$  reduces to $\b\Omega_{\n, 1}$ and one is left with the Lagrangian for Weyl gravity alone. This shows, we would argue, that the electromagnetic field cannot be introduced in the way proposed in \cite{Merkulov1984_II} because the underlying geometry is too rigid, so to speak.

\section{Conclusion} 
\label{Conclusion} 

Tractors and twistors are frameworks devised to deal with conformal calculus on manifolds. Whereas it has been noticed that both are vector bundles associated to the conformal Cartan principal bundle endowed with its normal Cartan connection, it is often deemed more direct and intuitive to produce them, bottom-up, from the prolongation of defining differential equations, the Almost Einstein and Twistor equations respectively. In this paper  we have proposed a straightforward and top-down gauge theoretic construction of twistors via the dressing field method of gauge symmetries reduction. 
 Our scheme involves very little effort, and nothing beyond $2\times 2$ matrix multiplication.

 We started with the conformal Cartan gauge structure $\left ((\P, \varpi), \sE \right)$ over $\M$ with gauge symmetry given by the gauge group $\b\H$ - comprising Weyl $\W$, Lorentz $\SL$ and conformal boosts $\b \K_1$  groups - acting on gauge variables $\chi=\{\b\varpi, \b\Omega, \psi, \b D\psi\}$ which are the spin representation of the conformal Cartan connection, its curvature, a section of the  associated spin $\CC^4$-vector bundle $\sE$ and its covariant derivative.  

Applying the dressing field approach we showed that  $\b\K_1$-invariant composite fields $\chi_1$ could be constructed thanks to a dressing field $\b u_1$ built out of parts of the Cartan connection $\b\varpi$.
In particular, the dressed section $\psi_1 \in \sE_1$ was shown to be indeed a tractor and $\b D_1=d + \b\varpi_1$  a generalized tractor covariant derivative, the usual one being induced by the dressed normal conformal Cartan connection: $\b D_{\n, 1}=d +\b\varpi_{\n, 1}$. We have thus generalized the results of \cite{Friedrich77} linking the twistor connection to the normal conformal Cartan connection.

 Furthermore we stressed that, while the composite fields $\chi_1$ are genuine gauge fields w.r.t the residual Lorentz gauge symmetry $\SL$, they are gauge fields of a non-standard kind w.r.t the residual Weyl gauge symmetry $\W$. Such non-standard gauge fields, resulting from the dressing field method, implement the gauge principle of physics in a satisfactory way but are not of the same geometric nature than the fields usually underlying gauge theories. The trator bundle with connection $(\sE_1,\b D_{\n, 1})$, as a restriction of $(\sE_1, \b D_1)$, is then seen to be an instance of non-standard gauge structure over $\M$. 

The initial conformal gauge structure is encoded infinitesimally in the initial $\mathsf{BRST}$ algebra satisfied by the gauge variables $\chi$. As a general result the dressing field approach modifies the BRST algebra of a gauge structure. We then provided the new algebra $\mathsf{BRST}_{\ww, \l}$ satisfied by the composite fields $\chi_1$. 

We ended in showing that our approach allows to make clear the equivalence between Weyl gravity  and the Yang-Mills gauge theory of the twistor connection, first observed in \cite{Merkulov1984_I}. We also formulated a critic of the twistorial unification of Weyl gravity with electromagnetism proposed in \cite{Merkulov1984_II}, showing that the underlying geometry was too rigid to allow a unification thus conceived.

We concluded that our scheme of symmetry reduction giving us twistors had reached a final point since no dressing field was available for   further erasing  either the Lorentz Symmetry - as is possible in the tractor case thanks to the vielbein as dressing field - or the Weyl symmetry. But actually, one can conceive a conformal gauge theory involving both twistors and tractors, where the Weyl symmetry is further reduced thanks to a dressing field extracted from a tractor field, which then plays the role of a Higgs field. One would obtain a Lorentz-gauge theory where Dirac spinors and operator could be extracted from invariant twistors and twistor derivative respectively.
This will be investigated in a forthcoming  paper.

\section*{Acknowledgement}  

We wish to thank Serge \textsc{Lazzarini} and Thierry \textsc{Masson} (CPT Marseille) for their encouragements and  for supporting discussions while this work was under completion.

{
\small
 \bibliography{Biblio}
}

\end{document}